\RequirePackage{fix-cm}
\documentclass[envcountsect,numbook,natbib]{svjour3_arxiv4}   
\smartqed  

\usepackage{graphicx}
\graphicspath{{./figures/}}
\usepackage[utf8]{inputenc}
\usepackage{dsfont}

\usepackage{color}
\usepackage{url}
\usepackage{enumerate} 
\usepackage{booktabs}
\usepackage{tabularx}
\usepackage{ltablex}
\usepackage{mathptmx}
\usepackage{bm}
\usepackage{amssymb,amsmath,amsfonts,latexsym,amsbsy} 
\usepackage{longtable} 
\usepackage{textcomp}
\usepackage{lscape}

\usepackage[english]{babel}

\usepackage[T1]{fontenc}
\usepackage{helvet}
\usepackage{etoolbox}

\usepackage{booktabs}
\usepackage{xcolor} 
\usepackage[colorlinks=black, citecolor=cyan]{hyperref}
\usepackage[font=footnotesize,labelfont=bf]{caption}

\usepackage{scrextend}
\usepackage{graphicx}
\usepackage{mathtools}
\mathtoolsset{showonlyrefs}

\usepackage{subfig}
\usepackage{calrsfs}
\DeclareMathAlphabet{\pazocal}{OMS}{zplm}{m}{n}
\let\mathcal\pazocal
\usepackage{enumitem}
\setlist[enumerate]{leftmargin=*}

\usepackage[12pt]{extsizes}
\usepackage{geometry}
\usepackage{ifthen}
\allowdisplaybreaks

\spnewtheorem{theorem}{Theorem}{\bfseries}{\itshape}
\spnewtheorem{corollary}[theorem]{Corollary}{\bfseries}{\itshape}
\spnewtheorem{lemma}[theorem]{Lemma}{\bfseries}{\itshape}
\spnewtheorem{proposition}[theorem]{Proposition}{\bfseries}{\itshape}
\spnewtheorem{definition}[theorem]{Definition}{\bfseries}{\itshape}
\spnewtheorem{remark}[theorem]{Remark}{\bfseries}{\upshape}
\spnewtheorem{assumption}[theorem]{Assumption}{\bfseries}{\itshape}
\counterwithin{theorem}{section}
\smartqed

\usepackage{tikz}

\def \R{\mathbb{R}}               
\def \N{\mathbb{N}}               
\def \1{{\bf 1}}                
\def \0{{\bf 0}}
\def\qed{\hfill$\Box$}
\def\covid{Covid-19 }
\def\covidc{Covid-19}

\def \trace{\operatorname{tr}}

\definecolor{myred}{rgb}{0.9,0,0}  
\definecolor{mygreen}{rgb}{0,0.7,0}  
\definecolor{myblue}{rgb}{0.2,0,0.8}  
\definecolor{orange}{rgb}{1,0.6,0}  
\definecolor{olive}{rgb}{0.5,0.5,0}  
\definecolor{mylila}{rgb}{0.8,0.5,0.2}  
\definecolor{mygrey}{rgb}{0.6,0.6,0.6}  
\definecolor{mybrown}{rgb}{0.65,0.16,0.16}  
\definecolor{mymaroon}{rgb}{0.11,0.0,0.0}

\def \one{\mathcal{I}}

\newcommand{\condvar}{{Q}}  
\newcommand{\condmean}{M}  
\newcommand{\condvarEKF}{{\widetilde{Q}}}  
\newcommand{\condmeanEKF}{{\widetilde{M}}}  

\newcommand{\PoisRate}{\lambda}  
\newcommand{\TimCh}{\tau}  
\def\var{\operatorname{Var}}
\def\cov{\operatorname{Cov}}
\newcommand{\Var}{\var}

\newcommand{\Noise}{{\mathcal{B}}}

\newcommand{\norm}[1]{\|#1\|}
\newcommand{\Zpath}{{\mathcal{Z}}}
\newcommand{\Zpathn}{{\Zpath_n}} 
\def\cf{{ \color{black} \widetilde{f}_n}}
\def\ch{{ \color{black} \widetilde{h}_n}}
\def\csigma{{\widetilde{\sigma}_n}}
\def\cg{{\widetilde{g}_n}}
\def\cell{{\widetilde{\ell}_n}}
\def\cb{{\widetilde{b}_n}}

\def\driftX{f_X}
\def\diffX{\sigma_X}
\def\dfc{{U}}
\def\dfcmean{{\condmean}_{\dfc}}
\def\dfcvar{{\condvar}_{\dfc}}
\def\Fprior{{\mathcal{F}^I_0}}

\def\L{L}
\def\LV{\L^V}
\def\LR{\L^R}
\def\K{d}  
\def\KV{\K^V}
\def\KR{\K^R}
\def\PP{P}
\def\PV{\PP^V}
\def\PR{\PP^R}
\def\psiV{\psi^V}
\def\psiR{\psi^R}

\def\Vin{V^\text{in}}
\def\Vout{V^\text{out}}
\def\Rin{R^\text{in}}
\def\Rout{R^\text{out}}
\def\Count{\Theta}

\date{}    
\usepackage{algorithm}
\usepackage{algpseudocode}
\usepackage[algo2e]{algorithm2e} 
\begin{document}
	
	\title{Estimating Unobservable States  in Stochastic Epidemic Models with Partial Information}
	
	\titlerunning{Estimating Unobservable States  in Stochastic Epidemic Models}        
	
	\author{Florent Ouabo Kamkumo  \and Ibrahim Mbouandi Njiasse \and Ralf Wunderlich}
	
	\authorrunning{F.~Ouabo Kamkumo, I.~Mbouandi Njiasse, R.~Wunderlich} 
	
	\institute{Brandenburg University of Technology Cottbus-Senftenberg, Institute of Mathematics, P.O. Box 101344, 03013 Cottbus, Germany; \\ 
		\email{\texttt{Florent.OuaboKamkumo / Ibrahim.MbouandiNjiasse / ralf.wunderlich@b-tu.de}}  
	}
	
	\date{Version of  \today}

	\maketitle
	\begin{abstract}
		 This article investigates stochastic epidemic models with partial information and addresses the estimation of current values of not directly observable states. The latter is also called nowcasting and related to the so-called ``dark figure'' problem, which concerns, for example, the estimation of unknown numbers of asymptomatic and undetected infections. 
			
		The study is based on Ouabo Kamkumo et al.~\cite{kamkumo2025stochastic}, which provides detailed information about stochastic multi-compartment epidemic models with partial information and various examples. Starting point is a description of the state dynamics by a system of nonlinear stochastic  recursions resulting from a time-discretization of a diffusion approximation of the underlying counting processes. The state vector is decomposed into an observable and an unobservable  component. The latter is estimated from the observations using the extended Kalman filter approach in order to take into account the nonlinearity of the state dynamics. Numerical simulations for a \covid model with partial information are presented to verify  the performance and accuracy of the estimation method.	
	\end{abstract}	
	\keywords{Stochastic epidemic model \and    Diffusion approximation \and Partial information \and Extended Kalman filter \and \covid \and  Estimation of dark figures}		
	\subclass{
		     92D30 
		\and 92-10 
		\and  60J60 
		\and 60G35  
		\and 62M20 
	}	
	\newpage
	
	\setcounter{tocdepth}{2}
	\tableofcontents
	
	\newpage
	
\section{Introduction }
Mathematical epidemic models play an important role in predicting, controlling, and even eradicating infectious diseases. They describe the dynamics of the behavior of a particular disease spreading in a population, such as the recent \covid pandemic.  Since the course of an epidemic is influenced by various uncertainties, such models should include stochastic components to account for both forecast uncertainties and  nowcast uncertainties.  While the former are related to unpredictable fluctuations in the future course of the epidemic, nowcast uncertainties refer to the inability to accurately capture all components of the current state of the epidemic. The description of the current status of an epidemic usually suffers from so-called ``dark figures'', for example, due to unreported or undetected infections. There is often a significant discrepancy between the actual number of infections and the reported or confirmed cases. This is due to people with mild symptoms not seeking medical attention and screening, or asymptomatic infections not presenting any symptoms at all. 	

Furthermore, it is known from the \covid pandemic, see \cite{CharpentierElieLauriereTran2020}, that testing restrictions, such as regional differences in availability, capacity, and accessibility, also contribute to underreporting. Delays in data collection, administrative problems, and varying screening criteria, such as prioritizing symptomatic individuals, can further distort the number of reported cases. In addition, social stigma, fear of isolation, or concerns about the consequences of a positive diagnosis may discourage people from getting tested or reporting their symptoms, further exacerbating underreporting.	

 Nowcast uncertainties complicate the estimation of important epidemic measures (e.g. effective reproduction rate, infection prevalence) and hinder epidemic management. Reducing the impact of these uncertainties is crucial to gain a more accurate insight into the actual spread of a disease, provide guidance for the public health response and implement effective control measures. Developing epidemic models that explicitly account for the dark figures, can play a crucial role in estimating and subsequently reducing the influence of undetected infected individuals among a given population. This approach is particularly relevant in the context of managing infectious diseases like the \covid pandemic.

In this article, we focus on such nowcast uncertainties and the estimation of dark figures, which represent the unobservable or hidden part in the description of the epidemic course.
The study is based on our paper \cite{kamkumo2025stochastic}, in which we explain in detail the mathematical modeling approach that leads to stochastic epidemic models with partial information and discuss various examples. We use the results of this paper in \cite{njiasse2025stochastic} for the study of
social planner's decision-making problems  to achieve cost-effective containment of an epidemic. These problems are treated as a stochastic optimal control problems under partial information, and solved by dynamic programming techniques.

\paragraph{ Literature Review on General Filtering Methods}
 Statistical methods such as filtering for estimating unobservable  states are based on stochastic  models with partial observation, which divide the states into observable and latent (unobservable) states. O'Neill and Roberts \cite{o1999bayesian} address the problem of the frequent absence of data concerning the infection process when analyzing epidemic models. They use the Markov chain Monte Carlo (MCMC) approach in a Bayesian framework to infer missing data and unknown parameters of interest. Britton and O'Neill \cite{britton2002bayesian} proposed a method for estimating the infection rate and the average number of social contacts of an individual based on a random graph model that behaves like an SIR model. They developed an MCMC method to facilitate Bayesian inference for the parameters of both the epidemic model and the underlying unknown social structure.  Calvetti et al.~\cite{calvetti2021bayesian}  consider a modified SEIR (susceptible-exposed-infected-recovered) model to describe the \covid pandemic and propose an estimation method for the  temporal evolution of the unknown state and model parameters based on noisy observations of new daily infections. Their estimation approach uses a Bayesian particle filtering algorithm. Similarly,  Lal  et al.~\cite{lal2021application} consider a SIRD (susceptible, infected, recovered, deceased) model to describe the \covid pandemic. They then used the ensemble Kalman filter (EnKF) to estimate both the model parameters and the unobservable state. 
 Colaneri et al.~\cite{colaneri2022invisible} consider discrete-time stochastic SIR model, where the transmission rate and the true number of infectious individuals are random and unobservable. They follow a hidden Markov model (HMM) approach and apply nested particle filtering approach to estimate the  reproduction rate and the model parameters.
 Alyami  and Das  \cite{alyami2023extended} address the issue of outliers (unusual or extreme values in datasets) in \covid data that can lead to inaccurate estimates using traditional Gaussian Kalman filtering methods. The authors use a skew Kalman filter  that accounts for asymmetry in relevant quantities such as the distribution of the initial estimate, leading to a Bayesian inference method for state estimation.

The Bayesian framework for filtering is also considered \cite{ionides2011iterated,lekone2006statistical,streftaris2004bayesian}. It  is a powerful and flexible approach to statistical modeling and inference. However, as it is primarily a computer-based method, it does not provide  closed-form expressions for computing estimates.  MCMC, particle filtering or similar approaches used to approximate the posterior distribution can be slow to converge and computationally demanding, especially for large datasets or complex models. 

\paragraph{ Literature Review on Kalman Filtering}
The article Zhu  et al.~\cite{zhu2021extended} develops a stochastic SEIR(R)D-SD model (susceptible, exposed, infectious, recovered (re-infected), deceased, social distancing)  to model the dynamics of \covidc. The model takes into account immunity loss rates and social distancing factors to account for the uncertainties associated with the spread of \covidc. An extended Kalman filter (EKF) was used to estimate model parameters and transmission status, improving prediction accuracy. Sebbagh  and Kechida \cite{sebbagh2022ekf} used the EKF and  applied it to a  SIRD model  to predict the spread of \covid in Algeria. This method made it possible to predict daily infection, mortality and recovery rates, as well as basic reproduction numbers, thus contributing to effective pandemic management. Hasan  et al.~\cite{hasan2022new} proposes a statistical approach based on the SIRD model to describe the \covid pandemic.  They use an EKF to estimate parameters dynamically, providing insight into disease progression. Zeng and Ghanem  \cite{zeng2020dynamics} have used a switching Kalman filter  formalism based on a linear Gaussian model to apply dynamical learning and prediction to the new \covid daily cases.  Due to its design, this filter is more effective at estimating the hidden states of processes whose underlying dynamics are non-linear and non-Gaussian, which is often the case in practical applications. However, managing multiple models and switching mechanisms increases computational complexity, while the correct tuning of probabilities and transition mechanisms can prove difficult. In \cite{njiasse2025stochastic}, Njiasse et al. use the EKF to solve stochastic optimal control problems under partial information that arise in connection with the decision-making problems of a social planner seeking to contain an epidemic in a cost-efficient manner.  Chen et al.~\cite{chen2022conditional} explores the capabilities of a rich class of nonlinear stochastic models, known as the conditional nonlinear Gaussian system (CGNS), as approximate models for complex nonlinear systems. These models have several distinctive features. They allow the development of fast algorithms for the simultaneous estimation of parameters and unobserved variables, with uncertainty in the presence of partial observations. In addition, the associated conditional Gaussian distribution can be computed using closed-form analytical formulas, which greatly facilitate the mathematical analysis and numerical simulations of CGNS models. In fact, closed analytical formulas for conditional Gaussian distributions enable the development of efficient and statistically accurate algorithms for parameter estimation, data assimilation, and nowcasting uncertainties when only partial observations are available. In Chen and Majda \cite{chen2018conditional}, the authors consider a conditional Gaussian framework to understand and forecast complex multiscale nonlinear stochastic systems. They emphasize that such systems can effectively capture the non-Gaussian characteristics inherent in natural phenomena.

\paragraph{Our Contribution} 
 
This article derives stochastic epidemic models from large populations limits of microscopic models that are based on continuous-time Markov chains. These models explicitly take into account unobservable states of an epidemic and enable the application of filtering methods to estimate dark figures.  Further, we improve and extend traditional stochastic models by including a cascade of states to account for partially hidden compartments in which either inflow or outflow can be observed, but not both at the same time. This allows to capture all available information for the estimation process. A particular focus is on \covid models, in which the most important parameters have been calibrated using German \covid data. In order to account for nonlinearities in the state dynamics, we apply the extended Kalman filter to estimate unobservable states.
We conduct extensive numerical simulation studies for the COVID model to evaluate the performance in estimating undetected infections. They show that, despite initially rather inaccurate estimates, the filtering processes learn very quickly from observations and can track unobservable states with high accuracy, thereby supporting public health measures.  

\paragraph{Paper Organization}
In Section \ref{sec:model}, we introduce stochastic epidemic models with partial information and discuss two \covid models in more detail. Section \ref{sec:filter} is devoted to  the Kalman filtering approach for the estimation of unobservable states of epidemic models, and introduces the extended Kalman filter method. This method is applied to the \covid models in Section \ref{sec:application}, which also contains details on the calibration of the initial estimates of the Kalman filter. Finally, Section \ref{sec:NumericalResults} presents the results of extensive simulation studies that demonstrate the performance of the filter estimates.

\paragraph{Notation} Throughout this paper, we use the notation $x^1,\ldots,x^d$ for the entries of a vector $x\in \R^d$,  and $\norm{x}$ is the Euclidean norm of $x$. The entries of a matrix $A$ are denoted by $A^{ij}$,  and  $\norm{A}$ denotes the Frobenius norm of $A$. The identity matrix in $\R^{d\times d}$ is denoted by $\one_d$, and $0_{d}$ denotes the null vector in $\R^d$.

\section{ Epidemic Models}
\label{sec:model}
 In this section, we sketch the derivation of stochastic epidemic models in Subsection \ref{sec:epidemic_model}, starting from a microscopic level in which Poisson counting processes describe the state of an epidemic. By examining such models for a growing total population size and applying functional limit theorems from stochastic process theory, diffusion approximations are derived. These describe the dynamics of epidemics through systems of stochastic differential equations. Finally, time discretization leads to stochastic recursions, which are used in the further course of the article. They are modified by dividing the state vector into an observable and an unobservable or hidden component leading to a model with partial information. 

In Subsection \ref{sec:BaseModel}, we propose a mathematical model for the dynamics of the \covid pandemic. It takes into account important features that have been observed during the pandemic. One of these is the high proportion of asymptomatic patients and the low testing rate in many countries. Another feature is the category of people who develop symptoms of the disease and are unofficially confirmed as positive (e.g., by rapid tests from local stores) but refuse to go to a testing center for official confirmation for various reasons. The model developed is referred to as the ``base model''. 

Among the hidden states, there are some ``partially hidden states'' that describe the subpopulation in compartments with an observable inflow representing the recovery of infected or vaccinated individuals. They are considered fully immune, but only for a certain period of time. However, for \covidc, it is known that after this period, immunity gradually wanes over time and individuals must be considered susceptible again. Since this transition is not reported, the outflow from these compartments remains hidden. To capture the information about the observable inflow, the basic model in Subsection \ref{sec:ExtendedModel} is refined into an ``extended model'' that includes additional cascades of new states that take into account the time elapsed since recovery or vaccination.

\subsection{Stochastic Epidemic Models with Partial Information}\label{sec:epidemic_model}
The setting is based  on \cite{kamkumo2025stochastic}. For self-containedness  and the convenience of the reader,   the most important components of the model are briefly summarized in this subsection. 
We consider a compartmental epidemic model in which a population of constant size $N\in\N$ is decomposed into $d\in\N$ subpopulations, which form the compartments, and $K\in\N$ transitions between the compartments.  
Let $T>0$ be a fixed horizon time, and $(\Omega,\mathcal{F},\mathbb{F},\mathbb{P})$ be a filtered  probability space  with the filtration $\mathbb{F}=\{\mathcal{F}_{t}\}_{t\in[0,T]}$, a family of $\sigma$-algebras with  $\mathcal{F}_s\subset \mathcal{F}_t\subset\mathcal{F}$ for $0\le s< t\le T$. Further, let $X=(X(t))_{[0,T]}$ be  a stochastic process  with values in $\{0,\ldots,N\}^{d}$, where  $X_{i}(t)\in \{0,\ldots,N\}$ denotes the absolute subpopulation size of   compartment $i=1,\ldots,d$ at time $t$.  The filtration $\mathbb{F}$ is assumed to be generated by $X$, that is the $\sigma$-algebras $\mathcal{F}_{t}=\sigma\{X(s),s\le t\}$  model the information available from observing $X$ in $[0,t]$.

\paragraph{Microscopic Models}
We denote by $\Count_{k}(t)$ the total number of transitions of type  $k=1,\ldots,K$, in the time interval $[0,t]$. 
Given these counting processes, the dynamics of $X$ can be expressed as
\begin{align}\label{stateX}
	X(t)=X(0)+\sum\limits_{k=1}^{K}\xi_{k} \, \Count_{k}(t), 
\end{align}
with the $d$-dimensional  transition vectors $\xi_k$, $k=1,\ldots,K,$ defined as the increment of the state process $\xi_k=X(t)-X(t-)$, if transition $k$ occurs at time $t$. Here, $X(t-)=\lim_{s\uparrow t} X(s)$ denotes the left limit of $X$ at time $t$,  that is the state immediately before the transition. Under the natural assumption that, during a transition, at most one individual moves from one compartment to another, the entries of these vectors only take the values $0$, $+1$ and $-1$, and it holds 
\begin{align}
	\xi_{k}^{i} & =\left\{\begin{array}{%
			l @{\qquad}  l}
		-1 & \text{if an individual leaves compartment $i$,}\\
		+1 & \text{if an individual enters compartment $i$,}\\
		\phantom{-}0& \text{otherwise.}\\
	\end{array}\right.
\end{align}
Another natural assumption, which is fulfilled in many applications, is that the transitions between the compartments are independent of each other. Thus, the counting processes $\Count_1,\ldots,\Count_K$ are assumed to be independent. Further, they are modeled as Poisson processes with intensities $\lambda_{k}=\lambda_{k}(t,X(t))$ depending  on time $t$ and the current state $X(t)$, as in    \cite{guy2015approximation,guy2016approximation,hasan2022new}.
Let $\Pi_1,\ldots,\Pi_K$ be independent standard Poisson processes with unit intensity, and consider the time change $\TimCh_{k}(t)=\int_{0}^{t}\PoisRate_{k}(s,X(s))ds$. Then, the non-homogeneous Poisson counting  processes $\Count_{k}$ can be expressed in term of the standard Poisson processes  $\Pi_{k}$  as  $\Count_{k}(t)=\Pi_{k}(\TimCh_{k}(t))$.
Thus, the state dynamics given in \eqref{stateX} can be expressed as
\begin{align}\label{stateX_Poisson}
	X(t)=X(0)+\sum_{k=1}^{K}\xi_{k} \Pi_{k}\Big(\int\nolimits_{0}^{t}\lambda_{k}(s,X(s))ds\Big).
\end{align}
The above  dynamics of the state process is refereed in the literature as continuous-time Markov chain (CTMC).

\paragraph{Macroscopic Models}
For large population sizes $N$ the study of the asymptotic behavior of the CTMC dynamics given in \eqref{stateX_Poisson} and the application of a functional  law of large numbers and  central limit theorem, see Britton and Pardoux ~\cite[Chapter 2, Section 2.2-2.3]{BrittonPardoux2019}, Anderson and Kurtz \cite[Chapter 1, Section 3.2]{anderson2011continuous}, Ethier and Kurtz \cite[Chapter 4, Section 7]{ethier2009markov},  Guy et al.~\cite{guy2015approximation}, and \cite[Section 3.4]{kamkumo2025stochastic}, results in the so-called diffusion approximation of $X$ by a diffusion process $X^D$ which solves the following system of stochastic differential equations (SDEs) 
\begin{align}\label{DA1}
	dX^D(t) & =\driftX(t,X^D(t))dt+\diffX(t,X^D(t))dW(t), \qquad  X^D(0)=X(0)=x_0,
\end{align}
driven by $K$-dimensional  standard  Brownian motion $W$. The drift and diffusion coefficient $f_X$ and $\sigma_X$  are given by 
\begin{align}\label{DA_drift_diffusion}
	\begin{split}
		\driftX(t,x) & =\sum_{k=1}^K \xi_k \lambda_k (t,x) \quad\text{and}\quad 
		\diffX(t,x) =\big(  \xi_1 \sqrt{\lambda_1(t,x)},\ldots,\xi_K \sqrt{\lambda_K(t,x)}\big).		
	\end{split}	
\end{align}
For the filtering approach, it is convenient to work with a discrete-time approximation of the state dynamics because the observable information (e.g., reported cases or compartment counts) is typically available at discrete time points (e.g., daily, weekly). We therefore divide the time interval  $[0,T]$ into $N_t\in \mathbb{N}$ uniformly spaced subintervals of length $\Delta t = {T}/{N_t}$ and define the time grid points  $t_n = n \Delta t$ for $n =0,\ldots,N_t$. 
Discretizing  SDE \eqref{DA_drift_diffusion} using the Euler-Maruyama scheme yields a stochastic recursion of the form:
\begin{align}\label{Discrete_Dyn1}
X^{D}_{n+1}&=X^{D}_n +  \driftX(n,X^{D}_{n})\Delta t + \diffX(n,X^{D}_{n})\sqrt{\Delta t}\,\Noise_{n+1}. 
\end{align} 
Here, $X^{D}_n$ denotes the discrete-time approximation of the state $X^D(t_n)$ at time $t_n$. Further, $(\Noise_n)_{n=1,\ldots, N_{t}}$ is a sequence of independent standard normally distributed Gaussian vectors, $\Noise_n\sim \mathcal{N}(0_{K},\one_{K})$.

\paragraph{Models with Hidden States} 
We now want to distinguish between observable and hidden states and assume that among the $d$-states there are $d_1$-states that are hidden, where  $d_1\in \N, d_1<d$. That is, their actual subpopulation size is not directly observable, but can only be estimated from the observable states. The latter are the remaining $d_2=d-d_1$ states. The $d$-dimensional state vector $X^{D}$ is then decomposed into a $d_1$-dimensional vector $Y$, whose entries are the hidden states, and a $d_2$-dimensional vector $Z$, which contains the observable states. The state dynamics in \eqref{Discrete_Dyn1} can then be rewritten in form of the following system of recursions 
\begin{align}\label{state_YZ}
	\begin{split}
		Y_{n+1}&= f(n,Y_n,Z_n) +~ \sigma (n,Y_n,Z_n){\Noise^1_{n+1}} 
		+ g(n,Y_n,Z_n)  \Noise^2_{n+1}\\[1ex]
		{Z_{n+1}}&=   h(n,Y_n, Z_n) + ~\ell(n,Y_n,Z_n){\Noise^2_{n+1}}
	\end{split}	
\end{align}
where $(\Noise^1_n)$ and  $(\Noise^2_n)$ are independent sequences of i.i.d.~$\mathcal{N}(0_{k_1},\one_{k_1})$ and $\mathcal{N}(0_{k_2},\one_{k_2})$ random vectors, respectively, so that $\Noise^1_n$ contains those $k_1\in\{1,\ldots,K-1\}$ entries of $\Noise_n$, that only appear in the recursion for $Y$, while $\Noise^2_n$ collects the $k_2=K-k_1$ entries of $\Noise_n$ that appear in the recursion for both $Y$ and $Z$. In the following two subsections and in Appendix \ref{app:coeff}, we specify the specific form of the coefficients $f,h,\sigma,g,\ell$ for two examples of epidemic models.

\subsection{Base \covid  Model}\label{sec:BaseModel}
We now propose a stochastic model of  the dynamics of the \covid pandemic that takes into account unobservable states.
This is motivated on the one hand  by the high proportion of asymptomatic patients, and the low rate of use of tests in many countries, and on the other hand by  individuals who develop symptoms of the disease and are confirmed to be positive for the disease unofficially (e.g.~through rapid tests obtained from local stores), but refuse to go to a testing center to confirm officially whether or not they are infected. The latter is thus a crucial problem in the fight to eradicate the disease, because if their numbers were to increase, we would be facing an outbreak of infectious individuals, which would compromise the government's efforts to eradicate the disease. We start in this subsection with  a ``base model''. It will be refined in Subsection \ref{sec:ExtendedModel} to account for so-called ``partially hidden states'' and that enables to capture more of the available  information contained in the time elapsed since recovery or vaccination.

\begin{figure}[h]
	\centering
	\includegraphics[width=10cm,height=7cm]{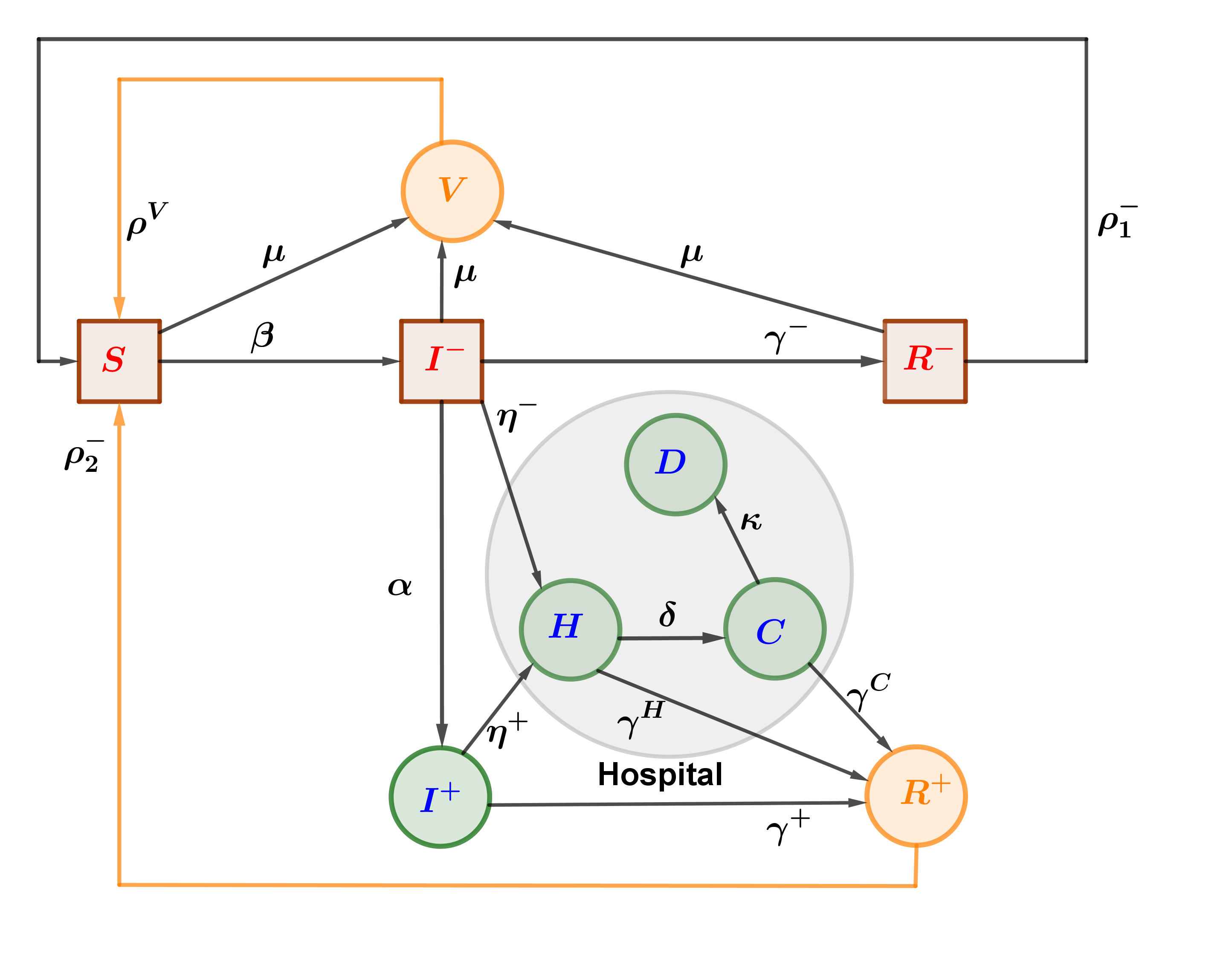}
	\caption{ Base \covid model with partial information consisting of of three fully hidden states \(I^-, R^-, S\), two partially hidden states \(R^+, V\), and four observable states \(I^+, H, C, D\).}
	\label{ch4_Model SIRpm_7}
\end{figure}

Starting point of the model depicted in Figure \ref{ch4_Model SIRpm_7} is the celebrated SIRS model that divides the population into three compartments containing  susceptible ($S$), infected ($I$) and recovered ($R$) individuals, with transitions from $S$ to $I$ (infection), $I$ to $R$ (recovery), and $R$ to $S$ (losing immunity). Here, we divide the infected compartment  $I $ into two. First,  the compartment  $I^{-}$ collects the  undetected infected individuals, most often asymptomatic but contagious, whereas $I^{+}$ contains the detected  infected individuals. 
When an infected individual  is tested positive, it transitions  from $I^{-}$ to $I^{+}$. Additionally, we also split the recovered compartment into two different compartments. First,  $R^{-}$ contains the undetected recovered,  and  $R^{+}$ the detected recovered. Here, recovered individuals from $I^{-}$ transition to $R^{-}$, while $R^{+}$ collects those who recovered from $I^{+}$.  Since recovery from \covid does not confer lifelong immunity, people in $R^-$ and $R^+$ can lose their immunity and transition to $S$, i.e. they become susceptible again.

\begin{table}[!h]
	
	\begin{center}
		{\footnotesize 
			\setlength{\tabcolsep}{5pt}
			\begin{tabular}{l|l|c|c}
				\hline
				k & Transition  &Transition vectors $\xi_{k}^\top$& Intensity $\lambda_{i}(t,X)$  \rule{0pt}{2.5ex}\\ 
				\hline
				&&&\\[-1em]
				1& Infection of  susceptible   & $(1,0,0,0,-1,0,0,0,0)$ &  $\beta S\frac{I^{-}}{N}=\beta Y^{5}\frac{Y^{1}}{N}$ 
				\\
				\hline
				&&&\\[-1em]
				2& Test of infected undetected   & $(-1,0,0,0,0,1,0,0,0)$ &  $\alpha I^{-}=\alpha Y^1$ 
				\\
				\hline
				&&&\\[-1em]
				3 & Recovering of infected detected   & $(0,0,1,0,0,-1,0,0,0)$ &  $\gamma^{+}I^{+} = \gamma^{+}Z^{1}$ 
				\\ 
				\hline
				&&&\\[-1em]
				4 & Recovering of  infected undetected   & $(-1,1,0,0,0,0,0,0,0)$ &  $\gamma^{-}I^{-} = \gamma^{-}Y^{1}$
				\\ 
				\hline
				&&&\\[-1em]
				5 & Losing immunity of detected recovered  & $(0,0,-1,1,0,0,0,0,0)$ &  $\rho^{-}_{2}R^{+} = \rho^{-}_{2}Y^{3}$ 
				\\ 
				\hline
				&&&\\[-1em]
				6 & Losing immunity of undetected recovered  & $(0,-1,0,1,0,0,0,0,0)$ &  $\rho^{-}_{1}R^{-} = \rho^{-}_{1}Y^{2}$ 
				\\ 
				\hline
				&&&\\[-1em]
				7 & Losing immunity of vaccinated   & $(0,0,0,-1,1,0,0,0,0)$ &  $\rho^{V}V = \rho^{V}Y^4$ 
				\\
				\hline
				&&&\\[-1em]
				8 & Vaccination of undetected infected  & $(-1,0,0,0,1,0,0,0,0)$ &  $\mu I^{-} = \mu Y^{1}$ 
				\\ 
				\hline
				&&&\\[-1em]
				9 & Vaccination of undetected recovered  & $(0,-1,0,0,1,0,0,0,0)$ &  $\mu R^{-} = \mu Y^{2}$ 
				\\ 
				\hline
				&&&\\[-1em]
				10 & Vaccination of susceptible  & $(0,0,0,1,-1,0,0,0,0)$ &  $\mu S = \mu Y^{5}$ 
				\\ 
				\hline
				&&&\\[-1em]
				11 & Hospitalization of undetected infected  & $(-1,0,0,0,0,0,1,0,0)$ &  $\eta^{-} I^{-} = \eta^{-} Y^{1}$ 
				\\ 
				\hline
				&&&\\[-1em]
				12 & Hospitalization of detected infected  & $(0,0,0,0,0,-1,1,0,0)$ &  $\eta^{+} I^{+} = \eta^{+} Z^{1}$ 
				\\ 
				\hline
				&&&\\[-1em]
				13 & Recovering from Hospitalization   & $(0,0,1,0,0,0,-1,0,0)$ &  $\gamma^{H} H = \gamma^{H} Z^{2}$ 
				\\ 
				\hline
				&&&\\[-1em]
				14 & Recovering from ICU   & $(0,0,1,0,0,0,0,-1,0)$ &  $\gamma^{C} C = \gamma^{C} Z^{3}$ 
				\\ 
				\hline
				&&&\\[-1em]
				15 & Transfer to  ICU   & $(0,0,1,0,0,0,-1,0,0)$ &  $\delta H = \delta Z^{2}$ 
				\\ 
				\hline
				&&&\\[-1em]
				16 & Death   & $(0,0,1,0,0,0,0,-1,1)$ &  $\kappa C = \kappa Z^{3}$ 
				\\ 
				\hline
			\end{tabular}
		}
		\caption{Transition vectors and transition intensities of the base \covid model. The  state process 
			$X=\binom{Y}{Z}$ is decomposed into 
			$Y=(I^-,R^-,R^+,V,S)^{\top},$  $Z=(I^+,H,C,D)^{\top}$, the  total number of states is $ d= 9 $, and the number of  transitions is $ K = 16 $.}
		\label{Table_Info_Model1}
	\end{center}
\end{table}

An important feature of a pandemic such as \covid is the problem of ``flattening the curve'' as discussed in \cite{ferguson2020impact,qualls2017community}, which means that one of the objectives of choosing a public health intervention should be to avoid excessive in demand in the healthcare system, and in particular in intensive care units. Therefore, we incorporate  compartments representing interactions of individuals with the healthcare system and accounting of their vaccination status.
First, we introduce the compartments  $H$ of hospitalized individuals, and $C$ for individuals requiring treatment in intensive care units (ICUs).  The $C$ compartment captures critically ill patients who require ICU support, reflecting the strain on critical care resources. Since recovery of individuals in $H$ and $C$ is reported, those individuals  transition to $R^+$. The $D$ compartment comprises the individuals who die from the disease. We assume here for simplicity that deaths due to \covid occur exclusively among individuals in the $C$ compartment.
Second, due to the availability of vaccines, we include the vaccinated compartment $V$.  It aggregates all individuals who have received at least one dose of the vaccine, regardless of the number of doses or the vaccine type administered. Individuals who receive the vaccine are  susceptible ($S$), infected but not detected, or asymptomatic ($I^{-}$),  or recovered but not detected ($R^{-}$). As none of the available \covid vaccines  offers lifelong immunity, individuals in $V$   may lose their immunity over time and become susceptible, that is they transition to $S$.

The diagram in Figure \ref{ch4_Model SIRpm_7} illustrates all possible transitions within the base model.
Hidden states are contained in the vector  $Y=(I^-,R^-,R^+,V,S)^{\top}$, and the observable states in   $Z=(I^+,H,C,D)^{\top}$. The total number of states is $ d= 9 $ and the total number of transitions is $ K = 16 $. Note that the hidden states are divided into the ``fully hidden states'' $I^-,R^-,S$, in which both the inflow and the outflow are not observable, and the ``partially hidden states'' $R^+,V$, in which the inflow is observable but not the outflow. A refined model in Subsection \ref{sec:ExtendedModel} shows how the information from the observable inflow can be made available for an improved estimation of the fully hidden states. 

The dynamics of this model can be described using the CTMC approach as in \eqref{stateX} and \eqref{stateX_Poisson}, where the transition vectors and intensities of the Poisson counting processes are given in Table \ref{Table_Info_Model1}. 
Studying the asymptotic behavior of this model for large populations leads to the diffusion approximation in form of a system of SDEs as in \eqref{DA1}, and time-discretization to the system of recursions  \eqref{state_YZ} for the states $Y,Z$ for which we provide the coefficients $f,h,g,\sigma,\ell$ in Appendix \ref{app:base_coeff}.

\subsection{Extended \covid Model with Cascade States}\label{sec:ExtendedModel}
We now consider the inclusion of partially hidden compartments. In the context of \covidc, these compartments contain  vaccinated individuals, and individuals who have recovered following quarantine, hospitalization, or ICU care. These compartments are characterized by an observable inflow of individuals, since vaccination and recovery  is reported, but a unobservable outflow due to the lost immunity and a transition to the susceptible compartment. The latter is usually not reported.
	For many diseases, in particular for \covidc, it is known that vaccination and recovery provides full immunity only for a certain known period of time after recovery, followed by a phase of waning immunity and eventually complete loss of immunity, at which point individuals must be considered susceptible.
	To incorporate the information about full immunity into the model and use it to estimate the dark figures, the idea is to take into account the ``vaccination age'' and ``recovery age'', that is the elapsed time since vaccination and  recovery, respectively. Further, we introduce subcompartments that comprises individuals with the same vaccination and recovery age. To do this, we assume that the vaccination and  recovery ages are  multiples of $\Delta t$ and that full immunity lasts for a period of time of length $\LV\Delta t$ after vaccination, and $\LR\Delta t$ after a recovery,  with some given  $\LV,\LR\in\N$. Further,  a sequence or ``cascade`` of new compartments is introduced that comprise individuals of the same vaccination and  recovery age as depicted in Figure \ref{ch4_Model SIRpm_1016}. 
	
	\begin{figure}[h]
		\centering
		\includegraphics[width=11cm,height=6cm]{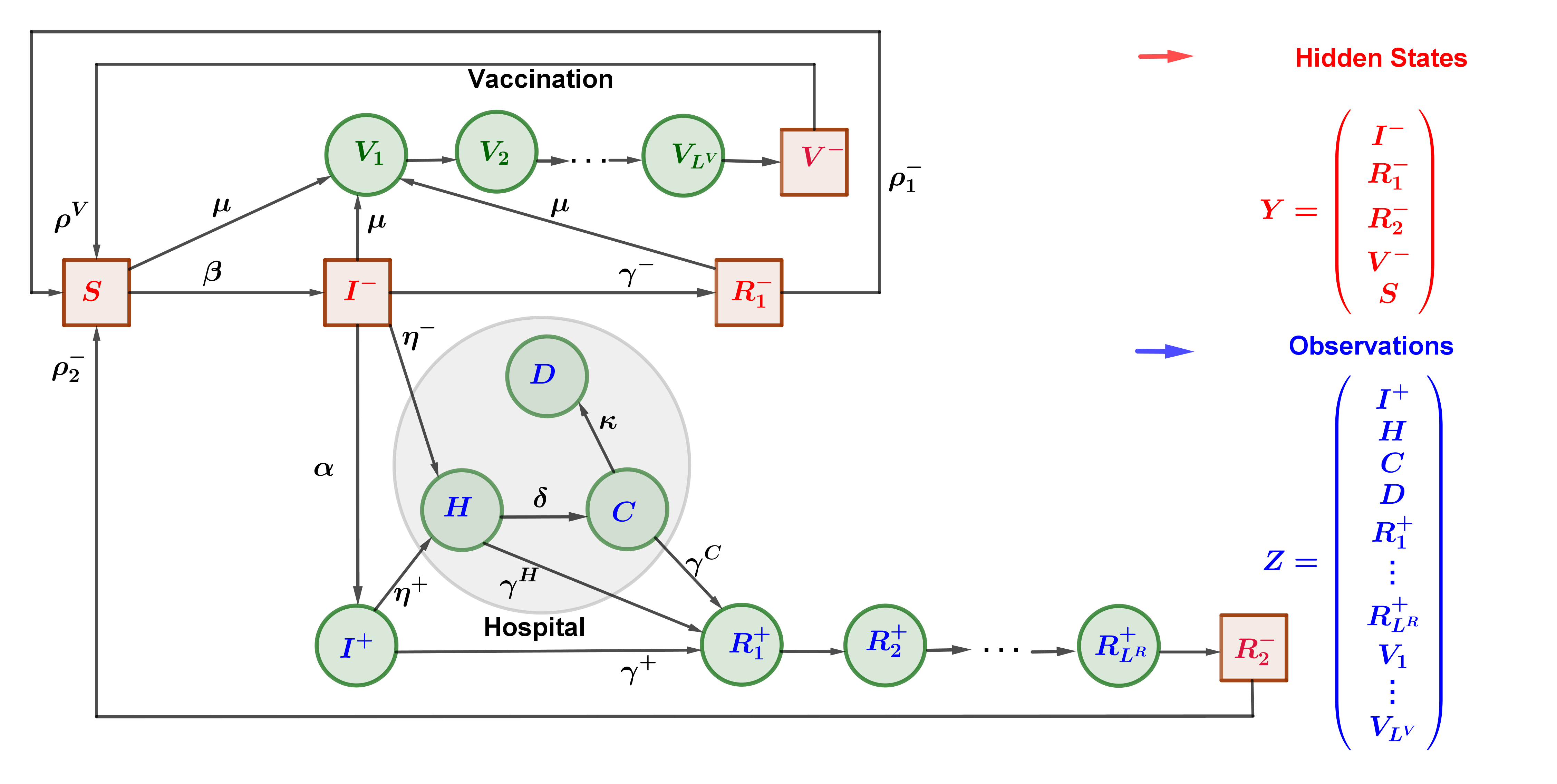}
		\caption{Extended Covid-19 model with partial information consisting of the 5 hidden states \(I^-, R^{-}_{1}, R^{-}_{2}, V^{-}, S\) and \(\KR + \KV + 4\) observable states. The four  observable states \(I^{+}, H, C, D\) were already present in the base \covid model depicted in  Figure \ref{ch4_Model SIRpm_7},  while $R^{+}_{1},\ldots, R^{+}_{\KR} $ and $V_{1},\ldots,C_{\KV}$ are cascade states.}
		\label{ch4_Model SIRpm_1016}
	\end{figure} 
	
	The approach is now  explained in more detail for the compartment $V$ of the vaccinated individuals. The procedure for $R^+$ is analogous.
	The compartment $V$ of the \covid base model is divided into the $\LV$ compartments $V_1,\ldots,V_{\LV}$ and $V^-$, where $V_j$ includes individuals with full immunity at vaccination age $j\Delta t$, $j=1,\ldots,\LV$, while $V^-$ includes individuals with a vaccination age exceeding $\LV\Delta t$ who no longer enjoy full immunity.
	The transition between the compartments  $V_1,\ldots,V_{\LV}$ is fully deterministic following the dynamics $V_{j,{n+1}}=V_{j-1,{n}}$ for $j=2,\ldots,\LV$. In the  first of the cascade compartments, the subpopulation size $V_{1,n+1}$ equals the random but observable inflow of vaccinated  $\Vin_n$ arriving from the compartments $S,I^-,R^-$  during the period  $(t_n,t_{n+1}]$. Thus it holds $V_{1,n+1}  =   \Vin_n$ where the inflow is given by $ \Vin_n = \mu(S_n+I^-_n+R^-_{1,n})\Delta t$.   Note that although $S,I^-,R^-_1$ are not observable, the inflow $\Vin$ is observable because vaccinations are reported. This information is recorded in $V_1$.

	The size of the subpopulation of compartment $V^-$ is hidden because a complete loss of immunity  and the transition of individuals to the susceptible compartment  cannot be observed. The dynamics reads $V^-_{n+1}=V^-_{n}+ V_{\LV,n} - \Vout_n$. Here,  $\Vout_n$ denotes the unobservable outflow of individuals that lose immunity and transition to $S$ during the period  $(t_n,t_{n+1}]$. It is given by 	$\Vout = \rho^V V^-_n \Delta t.$
	
	For large $\LV,\LR$ the above approach suffers from an  excessive number of new compartments which are added to the base \covid model.  Since the transitions between these cascade compartments are deterministic, observations of the associated subpopulation sizes provide little or no information for estimating the population sizes in the hidden compartments.  Only the first compartment $V_1$ receives a random inflow that carries information about the unobservable population sizes in $S,I^-$ and $R^-$. Therefore it appears reasonable to reduce the number $\LV$ of cascade compartments by combining adjacent compartments into one compartment. This results in some loss of information that could be gained from the observations, but allows for a smaller number of $\KV\le \LV$ additional compartments. Let $\PV_1,\ldots,\PV_{\KV}\in \N$ be the number of original cascade compartments that were combined into one compartment, with $\PV_1+\ldots+\PV_{\KV}=\LV$. Then $V_1$ comprises persons of vaccination age $1,\ldots,\PV_1$, $V_2$ comprises persons of vaccination age $\PV_1+1,\ldots,\PV_1+\PV_2$, and so on. In the aggregated compartments, individuals can no longer be distinguished based on their vaccination age. Therefore, the transition dynamics must be modified. We use an approximation based on the assumption that in compartment $V_j$ at each time $t_n$, a fraction of $\psiV_j=1/\PV_j$, $j=1,\ldots,\KV$, passes into the next cascade compartment. This approximation is accurate if the vaccination ages in $V_j$ are uniformly distributed, for example during a stationary phase of the epidemic. However, it may lead to some errors if there are rapid changes in the course of the epidemic. The above approach implies the following recursions 
	\begin{align}
		\label{cascade_dynamics_V}
		\begin{split}
		V_{1,n+1} & =  (1-\psiV_1)V_{1,{n}} +  \Vin_n ,\\	
		V_{j,{n+1}}& = (1-\psiV_j)V_{j,{n}} + \psiV_{j-1} V_{j-1,{n}},\qquad\quad  \text{for } j=2,\ldots,\KV,\\
		V^-_{n+1} &  = 	V^-_{n} + \psiV_{\KV} \;V_{\KV,{n}} -  \Vout_n. 
	\end{split}	
	\end{align}

	Similarly, the partially hidden compartment $R^+$ of the base \covid model can be first  divided into $\LR$ observable cascade compartments $R^+_1,\ldots,R^+_{\LR}$ where $R^+_j$ comprises  individuals of recovery age $j\Delta t$. The last cascade compartment $R^+_{\LV}$ is  followed  by an unobservable compartment $R_2^-$, including individuals with fading immunity and a recovery age greater than $\LR\Delta t$, see Figure \ref{ch4_Model SIRpm_1016}. Note that in the extended model we use the notation $R_1^-$ for the former $R^-$ compartment of the base model in order to distinguish from the new hidden compartment $R_2^-$. Here, it is assumed that after recovery from $I^+,H,C$, the full immunity lasts for  $\LR \Delta t$ units of time. Furthermore, the originally $\LR$ cascade compartments  are combined to $\KR\le \LR$ compartments by grouping $\PR_1,\ldots,\PR_{\KR}\in\N$ adjacent compartments, where  $\PR_1+\ldots+\PR_{\KR}=\LR$. 
	As above in \eqref{cascade_dynamics_V} we can the express the dynamics of this type of cascade states by the recursions
	\begin{align}
		\label{cascade_dynamics_R}
		\begin{split}
			R_{1,n+1} & =  (1-\psiR_1)R_{1,{n}} + \Rin_n ,\\	
			R_{j,{n+1}}& = (1-\psiR_j)R_{j,{n}} + \psiR_{j-1} R_{j-1,{n}},\qquad\quad  \text{for } j=2,\ldots,\KR,\\
			R^-_{2,n+1} &  = 	R^-_{2,n} + \psiR_{\KR} \;R_{\KR,{n}} -  \Rout_n. 
		\end{split}	
	\end{align}	
		Here, $\Rin_n$ denotes the observable inflow of recovered individuals and $\Rout_n$ the unobservable outflow of individuals that lose immunity and transition to $S$ during the period  $(t_n,t_{n+1}]$. They are given by 
		$\Rin_n= (\gamma^+ I^+_n + \gamma^H H_n + \gamma^C C_n)\Delta t\text{ and } \Rout = \rho^-_2 R^-_{2,n} \Delta t.	$

	The dynamics of the cascade compartments given in \eqref{cascade_dynamics_V} and \eqref{cascade_dynamics_R} can be incorporated into the system of recursions  \eqref{state_YZ} for the discrete-time dynamics of the state vectors  $Y,Z$ given in Figure \ref{ch4_Model SIRpm_1016}. In Appendix \ref{app:ext_coeff} we give  the coefficients of these recursions for the case $\KV=\KR=3$. This model is used for the numerical experiments reported below in Section \ref{sec:application} and \ref{sec:NumericalResults}.	
	For further details about the modeling with cascade states and including such compartments in continuous-time CTMC models  we refer to  \cite[Section 4.5]{kamkumo2025stochastic}.

\section{ Estimation of Unobservable States} 
\label{sec:filter}

In filtering theory, the objective is to estimate   at time $n$ the hidden state $Y_n$ of a dynamic system such as given in  \eqref{state_YZ}, based on the available past and current   observations $Z_0,\ldots,Z_n$,  and prior information about the distribution of the initial value of the hidden state. 
The goal is to determine the  mean-square optimal estimate of \( Y_n \).  
Mathematically, this translates into computing the conditional expectation  
\begin{align}
	\label{def_condmean}
	\condmean_n= \mathbb{E} \left[ Y_n \vert \mathcal{F}^Z_n  \right],
\end{align}
where the filtration \( \mathcal{F}^Z_n = \sigma \{ Z_k, k=0,\ldots,n \} \vee \Fprior \) represents the available information up to time \( n \), including prior knowledge about the distribution of the number of individuals in the hidden compartments, encoded in \( \Fprior \). Such prior information could originate from epidemiological reports, historical data, or expert assessments. The estimate \( \condmean_n \) is optimal in the mean-square sense, meaning it minimizes the expectation of the squared error   \(\mathbb{E} [\norm{Y_n - \hat{Y}_n}^2 ]\)
over all possible estimators \( \hat{Y}_n \) that are measurable with respect to \( \mathcal{F}^Z_n \). The accuracy of this estimate is described by the conditional covariance matrix 
\begin{align}
	\label{def_condvar}
	\condvar_n := \operatorname{Var}(Y_n|\mathcal{F}^Z_n) = \mathbb{E}[(Y_n - \condmean_n)(Y_n - \condmean_n)^\top \vert \mathcal{F}^Z_n ],
\end{align}
which quantifies the remaining uncertainty in the estimation of the hidden state.  The filter processes start at time $n=0$ with  the initial estimates  
\begin{align}
\label{def_filter_ini}
\condmean_0 = m_0 = \mathbb{E}[Y_0|\mathcal{F}^Z_0], \quad \text{and} \quad \condvar_0 = q_0 = \var(Y_0|\mathcal{F}^Z_0),
\end{align}
which are based on the prior information in $\Fprior$ and the first observation $Z_0$.

 A major challenge in the epidemic models considered above  is the nonlinearity present in both the drift $f$ of the hidden state process and the diffusion terms $g,\sigma,\ell$ governing the dynamics of the signal and observations. Unlike linear Gaussian systems, where Kalman filtering provides exact recursive solutions, such nonlinearity requires approximate filtering methods, such as the extended Kalman filter described in Subsections \ref{sec:EKF} and \ref{ACGS}. 
 It is based on the results of the next subsection.

\subsection {Kalman Filter for Conditionally Gaussian Sequences}\label{Sect_Cond_gaussian}
Let us consider the partially observable random sequence $\binom{Y_n}{Z_n}$ where $Y_{n}\in \R^{d_1}$ and $Z_{n}\in\R^{d_2}$ with   $n= 0,\ldots,N_t$, $N_t\in\N$, and $d_1,d_2\in \N$,  defined by the following stochastic recursions
\begin{align}
\begin{split}
	\label{kfcgm2}
	Y_{n+1} &=  \cf_0(\Zpathn)+\cf_1(\Zpathn) {Y_{n}}  +  \; \csigma(\Zpathn)\, \Noise^1_{n+1}+  \;\cg(\Zpathn)\,\Noise^2_{n+1},\\ 
	Z_{n+1}&=  \ch_0(\Zpathn)+\ch_1(\Zpathn) {Y_{n} } +\;\cell(\Zpathn)\,\Noise^2_{n+1}.
\end{split}	
\end{align}
for $n=0,\ldots,N_t-1$ and  initial values $Y_0,Z_0$. Here, $\Zpathn=(Z_0,\ldots,Z_n)$ denotes the path of the observations until time $n$, and $(\Noise^1_{n})_{n=1,\ldots,N_t}$,   $(\Noise^2_{n})_{n=1,\ldots,N_t}$ are independent sequences of standard normally distributed random vectors of dimension $k_1,k_2\in\N$, respectively.   The initial values $Y_0,Z_0$ are assumed to be independent of $(\Noise^1_{n}),(\Noise^2_{n})$. 
The vector-valued  functions $\cf_0,\ch_0$ and the matrix-valued functions $\cf_1,\ch_1,\csigma,\cg,\cell$   defined on $\R^{n+1}$ are measurable functions of  $\Zpathn$ such that the matrix products in \eqref{kfcgm2} are well-defined and of appropriate dimensions, 
and the following assumptions are satisfied.
\begin{assumption}~	\label{ass_filter}
	\begin{enumerate}[leftmargin=2.5em, itemsep=0.2ex]
		\renewcommand{\theenumi}{(A\arabic{enumi})}
		\renewcommand{\labelenumi}{\theenumi}
		\item\label{kfcgm_A3} 
		Let $\cb=\widetilde{b}(\Zpathn)$ be  any of the functions $\cf_0,\ch_0,\cf_1,\ch_1,\csigma,\cg,\cell$, then for all $ n=0,\ldots,N_t-1$
		$$\mathbb{E}[\norm{\cb(\Zpathn)}^2]<\infty.$$ 
		\item \label{kfcgm_A4}
		$\norm{\cf_1(\Zpathn)}$ and  $\norm{\ch_1(\Zpathn)}$ are $\mathbb{P}$-a.s. bounded for all $ n=0,\ldots,N_t-1$.
		\item \label{kfcgm_A5}
		$\mathbb{E}[\norm{Y_{0}}^2 + \norm{Z_{0}}^2]< \infty$.
		\item \label{kfcgm_A6}
		The conditional distribution $Y_{0}$ given $\mathcal{F}^{Z}_{0}$ is Gaussian.
	\end{enumerate}		
\end{assumption}
\begin{remark}
	Contrary to the system of equations \eqref{state_YZ} arising from modeling epidemics with not directly observable states, the drift term in the recursion for the signal sequence $(Y_n)$ in \eqref{kfcgm2} is linear in the signal while  in the base as well as extended \covid model the drift coefficient $f$ in \eqref{state_YZ} exhibits quadratic nonlinearities. Further, in \eqref{state_YZ}  the coefficients $\sigma,g,\ell$ scaling the random noise terms $\Noise_n^1,\Noise_n^2$ depend in a nonlinear way on the current values of the signal $Y_n$ and the observation $Z_n$, whereas in \eqref{kfcgm2} the corresponding coefficients $\cg,\csigma,\cell$ are independent of the signal but may depend on the whole observation path $\Zpathn$ until time $n$.
\end{remark}	
Given that the sequence $(Z_n)$ is observable while  $(Y_n)$ is unobservable, the filtering problem consists in constructing at each time $n=0,\ldots,N_t$ an estimate for the unobservable signal variable $Y_{n}$ based on the  observation path  $\Zpathn$ and the prior information encoded in $\Fprior$. As outlined in the introduction of this section, the signal estimate is desired in the form \eqref{def_condmean}, i.e., the conditional mean $M_{n}=\mathbb{E} [Y_{n}|\mathcal{F}_{n}^{Z}]$, accompanied by the conditional covariance $\condvar_n := \operatorname{Var}(Y_n|\mathcal{F}^Z_n) $, given in \eqref{def_condvar}, to quantify the estimation error. It is well-known that under Assumption \ref{ass_filter} the conditional mean $M_n$ is the mean-square optimal estimate of $Y_n$, see  \cite[Chapter 13]{LiptserShiryaevVolII2001}. Further,  $\trace(\mathbb{E}[Q_{n}])=\sum\nolimits_{i=1}^{d_{1}}\mathbb{E}[(Y_{n}^{i}-\condmean_{n}^{i})^{2}]$ yields the total estimation error.  

For a general recursive dynamics of the partially observable sequences $(Y_n),(Z_n)$, it is often quite tedious to determine the form of the conditional distribution of the signal $Y_n$ given the observation path $\Zpathn$ and its parameters $M_{n}$ and $Q_{n}$, and only possible using sophisticated numerical methods. However, for the dynamics given in \eqref{kfcgm2} with Gaussian drivers and initial conditions, and  drift terms that depend linearly on the hidden signal, closed-form solutions of the filtering problem in terms of recursions for  $(\condmean_{n})$ and $(\condvar_{n})$ become possible. They are based on the following result. 

\begin{theorem}[Liptser \& Shiryaev (2001)  \cite{LiptserShiryaevVolII2001}, Theorem 13.3]\label{thm_condGauss}\\
	Let Assumption  \ref{ass_filter} be satisfied. Then the sequence $\binom{Y_n}{Z_n}$ governed by \eqref{kfcgm2}  is conditionally Gaussian, i.e., the conditional distribution of $Y_n$ given $\mathcal{F}^Z_n$ is  is multivariate Gaussian  for any $n=0,1,\ldots,N_t$.
\end{theorem}

\begin{proof}
	A proof is given in \cite[Theorem 13.3]{LiptserShiryaevVolII2001}  and also in \cite{chen1989kalman}, and based on mathematical induction, where the authors  establish the normality of the conditional distribution of $Y_{n}$ given $ \mathcal{F}_{n}^{Z}$.  \qed
\end{proof}
The Gaussian nature of the sequences $(Y_n),(Z_n)$ enables to find the following  system  of recursive equations for the parameter sequences  $(\condmean_{n}), (\condvar_{n})$ of the conditional Gaussian distributions.
\begin{theorem}[Liptser \& Shiryaev (2001) \cite{LiptserShiryaevVolII2001}, Theorem 13.4]\label{Theo_EKF} \\ 
	Let Assumption  \ref{ass_filter} be satisfied and the sequences $({Y_n}),({Z_n})$ governed by \eqref{kfcgm2}. Then
	 the conditional distribution of \( Y_n \) given  \( \mathcal{F}^Z_n \) is the Gaussian distribution  $\mathcal{N}(\condmean_n,\condvar_n)$.	
	The conditional mean \( \condmean_n \) and conditional covariance \( \condvar_n \) are defined by the following recursions driven by the observations   
	\begin{align}		 
		\condmean_{n+1} =&   \cf_0 + \cf_1  \condmean_{n}   	+\\
		& \big( \cg\;\cell^\top + \cf_1 \condvar_{n} \ch_1^\top \big)  
		\big[\cell\;\cell^\top + h_1 \condvar_{n} \ch_1^\top\big]^{+}  		\big( Z_{n+1} - \big( \ch_0 + \ch_1 \condmean_{n} \big) \big),  \label{eq:EKF_mean_update} \\			
		\condvar_{n+1} = & -\big( \cg\;\cell^{\,\top} + \cf_1 \condvar_{n} \ch_1^\top \big)   
			\big[\cell\;\cell^\top + \ch_1 \condvar_{n} h_1^\top\big]^{+}  	\big( \cg\;\cell^\top+ \cf_1 \condvar_{n} \ch_1^\top  \big)^\top + \\
		& \cf_1 \condvar_{n} \cf_1^\top + \csigma\;\csigma^\top + \cg\;\cg^{\top}, 
		\label{eq:EKF_cov_update}
	\end{align}
	with the initial values \(\condmean_0 = m_0, \condvar_0 = q_0\).  All coefficient functions are evaluated at  $\Zpathn$.
	
\end{theorem}
\begin{proof}
	For the proof we refer to \cite[Theorem 13.4]{LiptserShiryaevVolII2001}, another reference is  \cite{chen1989kalman}.
\end{proof}

\begin{remark}
	 The notation \([A]^+\) represents the Moore-Penrose pseudoinverse of a matrix \( A \). This generalized inverse is particularly useful in filtering problems, as it ensures a well-defined solution even when \( A \) is singular or non-square. It satisfies the fundamental properties: \( A[A]^+ A = A \), \( [A]^+ A [A]^+ = [A]^+ \), \( (A [A]^+)^T = A [A]^+ \), and \( ([A]^+ A)^T = [A]^+ A \). Moreover, among all possible generalized inverses, \( [A]^+ \) provides the solution with the smallest Euclidean norm in least-squares problems, making it particularly suitable for stable state estimation in filtering and data assimilation applications. This minimal norm property is well established in linear algebra and optimization literature (see, e.g. \cite{ben2006generalized}). 
	
	 The pseudoinverse $	[\cell\;\cell^\top + \ch_1 \condvar_{n} \ch_1^\top]^{+}$
	is present in both the conditional mean and conditional covariance equations. Employing the pseudoinverse guarantees numerical stability, particularly in cases where the  coefficient \(\ell\) scaling the noise in the recursion for the observations is such that $\ell\ell^\top$  a singular matrix, as it is the case for the extended \covid model due to the cascade states.	
\end{remark}
\begin{remark}
	
	 Unlike the standard Kalman filter for linear Gaussian systems with constant coefficients, where the conditional covariance \( Q_n \) is computed by solving a  Riccati equation and can be determined offline i.e., prior to receiving any observations, the covariance computation for the Kalman filter for  conditionally Gaussian  sequences must be performed online. This is because, in this setting, the covariance \( Q_n \) may explicitly depend on the observation path \( \Zpathn \), requiring real-time updates as new data becomes available.
\end{remark}

\subsection{Extended Kalman Filter}
\label{sec:EKF}
In many real-life scenarios, systems exhibit nonlinear behavior, as illustrated in \cite{cazelles1995adaptive,ndanguza2017analysis,song2021maximum}, making the EKF essential to develop a filtering method capable of handling such complexities.  The EKF extends the applicability of the Kalman filter for conditional Gaussian sequences to nonlinear systems by linearizing the system dynamics at each time step, enabling an iterative estimation process. This adaptation addresses the limitations of the original Kalman filter and enables more accurate state estimation in a wider range of dynamic environments.\\
The model we will consider in our work and in particular in this section is described by the equations  \eqref{state_YZ} in Section \ref{sec:model}. The drift coefficient $f$ is non-linear with respect to the signal $Y$, and all diffusion coefficients $\sigma, g, \ell$ may also depend on the signal $Y$. With these  properties, the ensuing filtering problem becomes nonstandard due to nonlinearities in both the drift and the diffusion coefficients. Notably, the observation equation is influenced by signal-dependent noise. To address the nonlinearity in the drift coefficient $f$, we adopt the approach proposed by Gelb (1974) \cite{gelb1974applied}, Pardoux et al.~\cite{burkholder1991filtrage}, and Bain, Crisan \cite{bain2009fundamentals}, involving the linearization of $f$ through a Taylor expansion around a ``suitable'' reference point  $\overline{Y}_{n}$ in each time interval. Additionally, following \cite{picard1993estimation}, we freeze the diffusion matrix at the  reference point $\overline{Y}_{n}$ and can then apply the results for the Kalman filter for conditional Gaussian sequences  given in Theorem \ref{Theo_EKF}.
\subsection{Approximation by  Conditional Gaussian Sequences}\label{ACGS}
 In this section, we derive an approximation of the  nonlinear recursions \eqref{state_YZ} for the state processes $(Y_n),(Z_n)$ in the form of the recursions \eqref{kfcgm2} that will allow to apply the Kalman filtering results for  conditional Gaussian sequences. 	
	Our approach relies on linearizing the drift coefficient \( f \)  at each time interval with respect to the signal $Y_n$ using a first-order Taylor expansion around a suitably chosen reference point \( \overline{Y}_n \), which will be specified later.  Note that for the base and extended \covid model considered above, the drift coefficient  of the observation equation is already linear in the signal, and  does not require linearization. Therefore,  we will  restrict ourselves in the following to systems \eqref{state_YZ} with linear observation drift of the form 
		\begin{align}\label{h_linear}
			h(n,y,z)=h_0(n,z)+h_1(n,z)y.
		\end{align}  
	Denoting  the reference point for the expansion at time $n$ by $\overline{Y}_n=\overline y$, this results in the linearized coefficient  
\begin{align}
\label{linearization_f}
f(n, y, z) \approx f(n, \overline{y}, z) + \nabla_y f(n, \overline{y}, z)(y - \overline{y}),
\end{align}
where $\nabla_y$ denotes  the Jacobian  of the drift coefficient \( f \) with respect to \( y \).
Additionally, we replace  the signal  \( Y_n \) by the reference point \( \overline{Y}_n \) in the diffusion coefficients $g,\sigma,\ell$. This leads to the following  approximate recursion for  the state variables, denoted by \( \widetilde{Y}_n \) and \( \widetilde{Z}_n \).

\begin{lemma}\label{lem:linearized_system}
The approximation of the system dynamics \eqref{state_YZ} with a linear drift of the observation equation as in \eqref{h_linear}, after the first-order linearization   \eqref{linearization_f} of the drift coefficient $f$ and replacing the signal $Y_n$ with the reference  point  \( \overline{Y}_n \) in the diffusion coefficients $g,\sigma,\ell$, is given by the recursions
	\begin{align}\label{ekf_5}
		\begin{split}
			\widetilde{Y}_{n+1}&=  f_{0}(n, \overline{Y}_{n},\widetilde{Z}_{n})  + f_{1}(n,\overline{Y}_{n},\widetilde{Z}_{n})\widetilde{Y}_{n}+\sigma(n,\overline{Y}_{n},Z_{n}){\Noise^1_{n+1}} + g(n,\overline{Y}_{n},\widetilde{Z}_{n}) {\Noise^{2}_{n+1}},\\
			{\widetilde{Z}_{n+1}}&={ h_{0}(n,\widetilde{Z}_{n}) + h_{1}(n,\widetilde{Z}_n)\widetilde{Y}_{n} + \ell(n,\overline{Y}_{n},\widetilde{Z}_{n})}{\Noise^{2}_{n+1}}\\
			Y_{0}& =y, \quad Z_{0}=z. 
		\end{split}
	\end{align}
 The functions $f_0$ and $f_1$ are given for $\overline{Y}_n=\overline y$ by 
\begin{align}
	f_{0}(n, \overline{y},z) &= f(n,\overline{y},z)- \nabla_y f(n, \overline{y}, z)\overline{y} \quad \text{and} \quad 
	f_{1}(n,\overline{y},z)  = \nabla_y f(n, \overline{y}, z).
\end{align}		
\end{lemma}
 \begin{proof} Performing the linearization  \eqref{linearization_f} of $f$ and the substitution of $Y_n$ by the reference point in the in the diffusion coefficients $g,\sigma,\ell$ in \eqref{state_YZ} yields
	\begin{align}
		\widetilde{Y}_{n+1}&=   f(n,{\overline{Y}_{n}},\widetilde{Z}_{n}) +\nabla_y f(n,{\overline{Y}_{n}},\widetilde{Z}_{n})\big(\widetilde{Y}_{n}-{\overline{Y}_{n}}\big) + \sigma (n,{\overline{Y}_{n}}, \widetilde{Z}_{n})\Noise^1_{n+1} 		+ g(n,{\overline{Y}_{n}},\widetilde{Z}_{n})\Noise^2_{n+1}, 
		\label{ekf_3}\\
		\widetilde{Z}_{n+1}&= h_0(n,\widetilde Z_n)+h_1(n,\widetilde Z_n) \widetilde Y_n  +\ell(n,{\overline{Y}_{n}},\widetilde{Z}_{n})\Noise^2_{n+1}.  \label{ekf_4}			
	\end{align}
	Rearranging terms yields the assertion. \qed
\end{proof}

The linearized recursion in \eqref{ekf_5} can now be further transformed in the form of the recursions \eqref{kfcgm2} such that the Kalman filter results from Theorem \ref{Theo_EKF} can be used to derive approximations of the filter processes $(\condmean_n),(\condvar_n)$ for the original system \eqref{state_YZ}. These approximations will be denoted  by $(\condmeanEKF_n),(\condvarEKF_n)$. The key idea of the EKF approach is to choose the reference point $\overline{Y}_n $ for the linearization of the drift coefficient and the substitution of the signal $Y_n$ in the diffusion coefficients in each time interval  as the current (approximate) estimate $\condmeanEKF_n$ of the unobserved state $Y_n$. Further, the actual observation sequence  $(Z_n)$ is supposed to be generated by the recursion for $(\widetilde{Z}_n)$ in \eqref{ekf_5}. This leads to the recursive computation of the filter process approximations presented in Algorithm \ref{Algo_EKF}.

\begin{algorithm}[!ht] 
	\DontPrintSemicolon
	
	\KwIn{$Z_{0},\ldots,Z_{N_{t}}$; model parameters,   prior information $\Fprior$ } 
	\KwOut{Approximations $\condmeanEKF_n, \condvarEKF_n$ of $\condmean_n:=\mathbb{E}[Y_n|\mathcal{F}_{n}^{Z}]$ and $\condvar_n:=\var(Y_n|\mathcal{F}_{n}^{Z})$ for $n=0,\ldots,N_t$}
	\textbf{Initialization :}~~$n:=0$, \quad   $\condmeanEKF_0:=\condmean_0=\mathbb{E}[Y_0|\mathcal{F}_{0}^{Z}]$,  $\condvarEKF_0:=\condvar_0=\var(Y_0|\mathcal{F}_{0}^{Z})$ 
	\begin{enumerate}
		\item[(i)] State prediction 
		\begin{align}
			\begin{split}
				\condmeanEKF_{n+1}=   f_0+f_1  \condmeanEKF_{n}   ~+ ~\big( g\ell^\top+f_1 \condvarEKF_{n} h_1^\top \big)\big[\ell\ell^\top + h_1 \condvarEKF_{n} h_1^\top\big]^{+} \nonumber  \big( {\widetilde{Z}_{n+1}}- \big( h_0+h_1 \condmeanEKF_{n} \big) \big) 
			\end{split}
		\end{align}			
		\item[(ii)] Error measurement
		\begin{align}	
			\condvarEKF_{n+1}=  -\big( g\ell^\top+
			f_1 \condvarEKF_{n} h_1^\top \big)\nonumber \big[\ell\ell^\top + h_1 \condvarEKF_{n} h_1^\top\big]^{+}  \big( g\ell^\top+ f_1 \condvarEKF_{n} h_1^\top  \big)^\top 
			+f_1 \condvarEKF_{n}f_1^\top+ \sigma\sigma^T + gg^{\top} 
		\end{align}
		All coefficient functions are evaluated at the point $(n,\condmeanEKF_n, Z_n)$
		\item[(iii)] Repeat (i) and (ii) for the next time step until all samples
		are processed.
	\end{enumerate}
	
	\caption{EKF Algorithm \label{Algo_EKF}}
\end{algorithm}

Recalling that  $\Zpathn=(Z_0,\ldots,Z_n)$, the coefficients $\cb=\cf_0,\cf_1,\cg,\csigma,\cell$ of the filtering system  \eqref{kfcgm2} can be set to  $\cb(\Zpathn)=\cb((Z_0,\ldots,Z_n)) = b(n,\condmeanEKF_n,Z_n)$, while for the coefficients  $\cb=\ch_0,\ch_1$ appearing in the recursion  for the observations it holds $\cb(\Zpathn)=\cb((Z_0,\ldots,Z_n)) = b(n,Z_n)$. Note that $\condmeanEKF_n$ is defined recursively and its computation requires the knowledge on the complete observation path $\Zpathn$.

The recursion in Algorithm \ref{Algo_EKF} is initialized with the  mean and covariance $\condmean_0,\condvar_0$ of the conditional distribution of signal $Y_0$ at time $n=0$ given the prior information $\Fprior$ and the first observation $Z_0$. Recall, this distribution is assumed to be the Gaussian distribution $\mathcal{N}(\condmean_0,\condvar_0)$. The choice of these initial values is further discussed below in Section \ref{Initial_val}.

\begin{remark}\label{rem_error_analysis} 
For an analysis of the error resulting of the EKF approximation of the filter of the original nonlinear filtering problem we refer to the paper of Picard \cite{picard1991efficiency} and a recent extension by  Mbouandi Njiasse et al.~\cite{njiasse2025ExtendPicard}. There the authors justify in a  continuous-time setting  that under suitable regularity conditions on \( f \) and \( h \), the EKF provides a first-order approximation of the filter of the  nonlinear system, with an error depending on the smoothness of the model coefficients.
\end{remark}
\section{ Application to  \covid Models}\label{sec:application}
 In this section, the filter results mentioned above are applied to the base  \covid model presented in Subsection \ref{sec:BaseModel}  and depicted in Figure \ref{ch4_Model SIRpm_7}, and  the extended \covid model introduced in \ref{sec:ExtendedModel} with $\KR=\KV=3$ cascade states  and depicted in Figure \ref{ch4_Model SIRpm_1016}.  

\subsection{Linearization}\label{sec:applic_linear}
 The linearization of the nonlinear drift term $f$ that appears in the dynamics \eqref{state_YZ} of the hidden state $Y$ leads to  the linearized recursion in Equation \eqref{ekf_5} of Lemma \ref{lem:linearized_system} with the coefficients $f_0$ and $f_1$. They are given for the base  model in detail in  \eqref{drift_signal_linear_base} in  Appendix \ref{app:base_coeff}. The other coefficients $h_0,h_1,\sigma,g,\ell$ are not affected by the linearization, they are also given in  Appendix \ref{app:base_coeff}.
	
For the extended model the coefficients $f_0$ and $f_1$ are given in \eqref{drift_signal_linear_ext} in Appendix \ref{app:ext_coeff}. It also contains the details about  the other coefficients $h_0,h_1,\sigma,g,\ell$.

\subsection{Initialization of the Filter Estimates}\label{Initial_val}
Computing filter estimates requires the initialization of the filter processes 	$\condmeanEKF,\condvarEKF$  at time $n=0$. In the following, we show for the  \covid models introduced above, how $\condmeanEKF_0$ and $\condvarEKF_0$ can be constructed   based on the prior information contained in $\Fprior$, and the first observation $Z_0$. 
We give a detailed derivation for  the extended \covid model with cascade states as outlined in Subsection \ref{sec:ExtendedModel}, for which we provide numerical results in the next section.  The approach  for the base model introduced in Subsection  \ref{sec:BaseModel} is analogous, see Remark \ref{rem:ini_base_model}. Note that the initial values $\condmeanEKF_0,\condvarEKF_0$ are not yet distorted by linearization errors that occur in the EKF approximation of the “true” filter processes $\condmean_n,\condvar_n$ for $n>0$. Therefore, we remove the tilde from the notation and write $\condmean_0,\condvar_0$ instead of $\condmeanEKF_0,\condvarEKF_0$.

The extended \covid model consists of the five hidden states, $I^{-},R_{1}^{-}, R_{2}^{-}, V^{-}, S$ forming the vector $Y$, and ten observable states, $I^{+},H,C,D,R_{1}^{+}, R_{2}^{+}, R_{3}^{+}, V_{1}, V_{2}, V_{3}$ forming the observation vector $Z$. 
For the hidden states $R_{2}^{-}, V^{-}$, we assume for simplicity that they start with zero initial values. This seems to be an appropriate choice at the beginning of the pandemic, as these compartments will only accommodate the first individuals  after the phase of complete immunity following recovery and vaccination. Further, the vaccination was not yet available at the outbreak of the pandemic.

More challenging is the  estimation the numbers of  initially undetected infected  individuals $I^-_0$ and  undetected recovered individuals $R^-_{1,0}$, from which the estimation of the initially susceptible $S_0$ can then be derived from the normalization property using the assumption of a constant total population size $N$. Note that, unlike $R^-_{2}$, compartment $R^-_{1}$ already accommodates the first individuals at a relatively early stage, namely after the recovery of the first undetected infected individuals. Therefore, it cannot be assumed that it is empty at the start of the pandemic.

To estimate $I^-_0$ and $R^-_{1,0}$, it is helpful to work with  dark figure coefficients (DFC), which are defined as the ratio of undetected to detected individuals, i.e.  
\begin{align}\label{def:DFC}
	\dfc_{n}^{1}=\frac{Y_{n}^{1}}{Z_{n}^{1}} = \frac{I_{n}^{-}}{I_{n}^{+}} \quad \text{and}\quad \dfc_{n}^{2}=\frac{Y_{n}^{2}}{ Z_{n}^{5}+Z_{n}^{6}+Z_{n}^{7}} = \frac{R_{1,n}^{-}}{ R_{1,n}^{+}+ R_{2,n}^{+} + R_{3,n}^{+}}.
\end{align}
The advantage of DFCs is that they can be more easily quantified by analysts or experts and enable structured integration of prior knowledge into the estimation process. For example, interpreting $\dfc_{n}^{1}$ means that for every detected infected person, there are $\dfc_{n}^{1}$ undetected infected individuals, while $\dfc_{n}^{2}$ indicates the number of individuals who have recovered from an undetected infection per recovered individual with complete immunity.

At initial time $n=0$ we make the following 
\begin{assumption}\label{ass:dfc}~
\begin{enumerate}
\item 
Given the prior information	 $\Fprior$ the initial dark figure coefficients $\dfc_{0}^{1}$ and  $\dfc_{0}^{2}$ given in \eqref{def:DFC} are conditionally independent with conditional Gaussian distributions $\mathcal{N}(\dfcmean^i,\dfcvar^i)$, $i=1,2$.
\item
Initially, the compartments $R_{2}^{-}$ and $V^{-}$ are empty, i.e.,   $Y_{0}^{3} = R_{2,0}^{-}=0$ and $ Y_{0}^{4}= V^-_0 =0$.
\end{enumerate}
\end{assumption}	
Under this assumption the conditional mean $\dfcmean^i$  can be considered as an unbiased estimate of the initial DFC $\dfc_{0}^{i}$, $i=1,2$, or an ``expert's view'' provided by an analyst  equipped with the prior information encoded in $\Fprior$.  The accuracy of this estimate is described by the variance parameter  $\dfcvar^i$, in the sense that $(\dfcvar^i)^{1/2}$ is the the conditional standard deviation, while $1/(\dfcvar^i)^{1/2}$ serves as reliability measure. The smaller $\dfcvar^i$ the more accurate or reliable is the expert's view $\dfcmean^i$ for the unknown DFC $\dfc^i$. With $\dfcvar^i=0$ we can model full information about the initial value of the hidden state $Y_{0}^{1}$.

\begin{lemma}\label{lem:initial_estimate}
 For the extended \covid model and under Assumption \ref{ass:dfc}, the conditional distribution of the initial hidden state  \( Y_0 = (Y_{0}^{1},\ldots, Y_{0}^{5})^{\top} \)  given $\mathcal{F}_{0}^{Z}$, i.e., the prior information $\Fprior$ and the initial observation $Z_0$, is Gaussian. The mean and covariance matrix of this distribution are given by
\begin{align}\label{Mat_cov}
	\begin{split}	
\condmean_0= \mathbb{E}[Y_{0} | \mathcal{F}_{0}^{Z}]& =   (\dfcmean^{1}Z_{0}^{1},\dfcmean^{2}Z_{0}^{2},0,0,M_0^{Y_{5}})^\top, \\
\condvar_0 = \Var(Y_{0} | \mathcal{F}_{0}^{Z}) &=
\begin{pmatrix}
	\dfcvar^1( Z_{0}^{1})^2 & 0 & 0 & 0 & -\dfcvar^1( Z_{0}^{1})^2 \\[0.25ex]
	0 & \dfcvar^2( Z_{0}^{2})^2 & 0 & 0 & -\dfcvar^2( Z_{0}^{2})^2 \\[0.25ex]
	0 & 0 & 0 & 0 & 0 \\[0.25ex]
	0 & 0 & 0 & 0 & 0 \\[0.25ex]
	-\dfcvar^1( Z_{0}^{1})^2 & -\dfcvar^2( Z_{0}^{2})^2 & 0 & 0 & \dfcvar^1( Z_{0}^{1})^2 + \dfcvar^2( Z_{0}^{2})^2
\end{pmatrix},\\
\text{with}~~M_0^{Y_{5}}&=  N- M_0^{Y_{1}} - M_0^{Y_{2}}- \sum_{i=1}^{10}Z_{0}^{i}.
\end{split}
\end{align}

\end{lemma}

\begin{proof}
The proof is given in Appendix \ref{Proof_IE}.
\end{proof}

The structure of the covariance matrix $\condvar_0$ shows that  the susceptible population size \( Y_{0}^{5} \) is negatively correlated with the hidden compartments \( Y_{0}^{1} \) and \( Y_{0}^{2} \). This negative covariance arises because an increase in the number of undetected infections and recoveries reduces the number of susceptible individuals. Furthermore, the magnitude of these covariances depends on the initial uncertainty in the dark figure coefficients, represented by \( \dfcvar^{1} \) and \( \dfcvar^{2} \), scaled by the square of the respective observed detected cases \( Z_{0}^{1} \) and \( Z_{0}^{2} \).

\begin{remark}\label{rem:ini_base_model}
As the extended \covid model also the base model consists of  five hidden states, these are now $I^{-},R^{-}, R^{+}, V, S$. Again it is reasonable to assume that $V$ is empty at the beginning, as vaccination is not yet available. The estimates for $R^{-}, R^{+}$ are either deduced from expert view's or for simplicity one assumes that theses compartments are  initially  empty. The initial estimate for $I^-$ can be obtained as above in the extended model, and the estimate for $S$ is obtained using the normalization property. 

\end{remark}

\section{Numerical Results}\label{sec:NumericalResults}
In this section, we present and analyze numerical results from a simulation experiment in which we use an  implementation of the EKF Algorithm \ref{Algo_EKF} to estimate undetected states in a stochastic epidemic model with partial information. The focus is on the extended \covid model presented in Subsection \ref{sec:ExtendedModel}, for which we calibrate important model parameters using real data from the \covid pandemic in Germany. The goal is to evaluate the performance of the EKF filtering method and investigate how different model components and parameters influence the quality of the state estimation.

We begin in Subsection \ref{sec:paramters}, which provides an overview of the model parameters and their calibration. Subsequently, Subsection \ref{EKF_Performance} presents results for the performance of the estimation of undetected states, in particular undetected infected individuals, which are crucial for epidemic monitoring. Next, in Subsection \ref{Impact_IEtim}, we examine the sensitivity of the EKF to filter initialization by investigating how different choices for the initial conditional mean and covariance matrix affect the estimation accuracy.
Finally, in Subsection \ref{Impact_CS}, we examine the effects of cascade compartments containing individuals with complete immunity on the system dynamics and the estimation process.

\subsection{Settings for Numerical Simulations}\label{sec:paramters}
\begin{table}[h]
	
	\begin{center}
		\footnotesize 
		\begin{tabular}{ c |l| c | r }
			Parameters & Description & Value  & Reference \\ \hline
			$\beta$ & Transmission rate & \multicolumn{2}{l}{Time-dependent $\beta_n,\alpha_n,\mu_n$  } \\
			$\alpha$ & Test rate &\multicolumn{2}{l}{Fitted to German \covid data,} \\
			$\mu$ & Vaccination rate & \multicolumn{2}{l}{between 2020 and 2023, see \cite{kamkumo2025stochastic}} \\\hline			
			$\gamma^{-}$ & Recovery rate (undetected infected) & $ 1/14 $ & Based on \cite{bhapkar2020revisited}\\
			$\gamma^{+}$ & Recovery rate (detected infected) & $ 1/14 $ & Based on \cite{bhapkar2020revisited}\\
			$\gamma^{H}$ & Recovery rate from hospitalization & $ 0.048 $ & Based on \cite{CharpentierElieLauriereTran2020}\\
			$\gamma^{C}$ & Recovery rate from ICU & $0.02 $ & Assumed\\			
			$\eta^{+}$ & Hospitalized rate (detected infected) & $0.0023 $ & Based on \cite{CharpentierElieLauriereTran2020}\\ 
			$\eta^{-}$ & Hospitalized rate (undetected infected) & $0.0023 $ & Based on \cite{CharpentierElieLauriereTran2020}\\
			$\delta$ & Transfer rate to ICU (undetected infected) & $0.03 $ & Assumed\\
			$\kappa$ & Death rate  & $0.05 $ & Assumed\\
			$\rho^{-}_{1}$ & Rate of losing immunity (undetected recovered) & $ 120 $ & Assumed\\
			$\rho^{-}_{2}$ & Rate of losing immunity (detected recovered) & $  1/30 $ & Assumed\\
			$\rho^{V}$ & Rate of losing immunity (after vaccination) & $ 1/200 $ & Assumed\\
			\hline
			&&&\\[-1.5ex]
			$N$ &  Total population size  & $10^6$ & Assumed \\
			$T$ &   Horizon time & $3$ years & see \cite{kamkumo2025stochastic}\\
			$\Delta t$ &   Time step & $1$ day & Assumed, see \cite{kamkumo2025stochastic}\\
			$\LV=\LR$ &  Number of time steps with complete immunity & $90$ & Assumed, see \cite{kamkumo2025stochastic}\\
			$\KV=\KR$ &  Number of cascade compartments & $3$ & Assumed, see \cite{kamkumo2025stochastic}  \\
			$P_1=P_2=P_3$ &  Number of time steps grouped to one cascade comp. & $30$ & Assumed, see \cite{kamkumo2025stochastic}\\
			$\dfcmean^1=\dfcmean^2$ & Mean of dark figure coefficient & $10$ & Assumed\\			
			$\dfcvar^1= \dfcvar^1$ & Variance  of dark figure coefficient & $25$ & Assumed\\
			\hline
			
		\end{tabular}
	\end{center}
	\caption{Parameters values for the Extended  \covid model. 
	}
	\label{Table_Param_Covid19-1}
\end{table} 

 In Table \ref{Table_Param_Covid19-1}, we provide numerical values for the model parameters, in particular for the various parameters that control the transition intensities specified above in Table \ref{Table_Info_Model1}. The transmission or infection rate $\beta$, the test rate $\alpha$ and the vaccination rate $\mu$ are assumed to vary over time. They are calibrated to German COVID-19 data covering a three-year period from March 2020 to March 2023. The data is publically available\footnote{see \url{https://github.com/robert-koch-institut}}  from the Robert Koch Institute (RKI), a German federal government agency and research institute responsible for disease control and prevention. For the calibration procedure we refer to our article  \cite{kamkumo2025stochastic}.	The other parameters are assumed to be constants as given in Table \ref{Table_Param_Covid19-1}.
	
Based on these model parameters paths of both  the hidden and observable states have been generated using the recursion \eqref{state_YZ} and the initial values given in Table \ref{tab:init_vals}. 
Since no vaccination was available at the onset of the pandemic we start with empty compartments $V_1,V_2,V_3,V^-$. For simplicity we assumed that initially only the first cascade compartment $R_1^+$ including individuals recovered from a detected infection in the last $P_1=30$ days is non-empty, whereas $R_2^+,R_3^+,R_{2}^-$ containing individuals with larger recovery ages  are supposed to be  empty. Note that, unlike $R^-_2$, the compartment $R^-_1$ is not empty, as it contains individuals which recovered from a undetected infection,  but already from the first day after recovery.

The initial values for the filter processes $\condmean_0$ and $\condvar_0$ were determined by using Lemma \ref{lem:initial_estimate} and is based on Assumption \ref{ass:dfc}. In particular, we used the assumption that  the initial dark figure coefficients  $\dfc_1,\dfc_2$ for the ratios of undetected to detected numbers of infected and recovered individuals are independent normally distributed random variables with mean $\dfcmean^1=\dfcmean^2=10$ and variance  $\dfcvar^1=\dfcvar^2=25$, specified in Table \ref{Table_Param_Covid19-1}. 
Although, the actual values of the initial  DFCs $\dfc_{1,0}=I^-_0/I^+_0=300/75=4$ and $\dfc_{2,0}=R^-_0/(R^+_{1,0}+ R^+_{2,0}+R^+_{3,0})=80/50=1.6$, are located in the central region around the mean $\dfcmean^i=10$ of the supposed Gaussian distributions with standard deviation $(\dfcvar^i)^{1/2}=5,\, i=1,2$, the differences of the actual and estimated values of $I^-$ and $R^-_1$ are quite large. This is caused by the relatively large variances $\dfcvar^i$, which indicate a high degree of uncertainty in this estimate.
The initial estimate for $S$ is derived from the normalization property, that is the total population size minus the sum of all other initial compartment sizes and estimates, see Lemma \ref{lem:initial_estimate}.

After simulating the paths of the hidden and observable states, we applied the EKF algorithm to estimate the hidden states from the observations and evaluate the accuracy of the filter by comparing the estimated values with the actual simulated values. It is important to emphasize that the methodology is generic and can be adapted to models with other epidemiological parameters and total population sizes $N$.

\begin{table}[h!]
	\hspace*{0.1\textwidth}	
	\begin{tabular}{|c|c|c||c|c|}
		\hline
		\multicolumn{5}{|c|}{\textbf{Hidden States}}\\\hline
		$i$& {Variable} & {Initial value} & {Estimated initial value} & {Initial variance} \\
		& $Y^{i}$ & $Y_{0}^{i}$ & $\condmean_{0}^{i}=\condmean^{Y^i}_{0}$ & $\condvar^{ii}_{0}=\condvar^{Y^i}_0$ \\
		\hline
		$1$ & \( I^{-} \)      & 300  &  750 &  $375^2$  \\
		$2$ & \( R_1^- \)      &  200 & 500 &  $250^2$  \\
		$3$ &\( R_2 ^- \)      & 0  & 0         & 0  \\		
		$4$ & \( V^- \)      & 0 & 0 & 0  \\\hline
		$5$ &\( S \)      &  \multicolumn{3}{c|}{ by normalization} \\		
		\hline
	\end{tabular}
	\\[2ex]
	\hspace*{0.1\textwidth}	
	\begin{tabular}{|c|c|c||c|c|c|}
		\hline
		\multicolumn{6}{|c|}{\textbf{Observable States}}\\\hline
		$i$& {Variable} & {Initial value} & $i$& {Variable} & {Initial value}\\
		& $Z^{i}$ & $Z_{0}^{i}$& &  $Z^{i}$ & $Z_{0}^{i}$ \\
		\hline
		$1$ & \( I^{+} \)  & 75 & 6 &  \( R_2^+ \) &  0\\
		$2$ & \( H \)      & 15 & 7 &  \( R_3^+ \) &  0\\
		$3$ &\(C \)        & 10 & 8 &  \( V_1\)    &  0 \\		
		$4$ & \( D \)      &  5 & 9 &  \( V_2\)    &  0\\		
		$5$ &\( R_1^+ \)   & 50 &10 &  \( V_3\)    &  0\\		
		\hline
	\end{tabular}	
	\caption{Initial values of the state variables used in the simulation of paths for the extended \covid model.
		Top: Hidden states $Y_1,\ldots,Y_5$ together with initial estimates used for filtering. Bottom: Observable states.	}
	\label{tab:init_vals}
\end{table}

\subsection{Performance of  EKF  Estimation of Unobservable States}\label{EKF_Performance}
In this subsection simulated paths for the hidden and observable states of the extended \covid model are used to evaluate the proposed EKF estimate of the hidden states. The simulated data for the hidden states serve as reference values for comparison with the estimated values based on the EKF filter with the data from the observable states as input.

The simulation based on the model parameters and initial values given in Subsection  \ref{sec:paramters}  generated paths of all states   which are displayed 
in Figure \ref{fig:f170}. As a result of the calibration of the key parameters $\beta,\alpha,\mu$ to German \covid data, the paths for the   infections $I^{+}$ and $I^-$, exhibit a similar pattern over time to that observed during the \covid pandemic in Germany. 

\begin{figure}[!tbp]
	\centering
	\subfloat[Detected infected, undetected infected, and undetected recovered.] {\includegraphics[width=0.48\textwidth,height=0.26\textwidth]{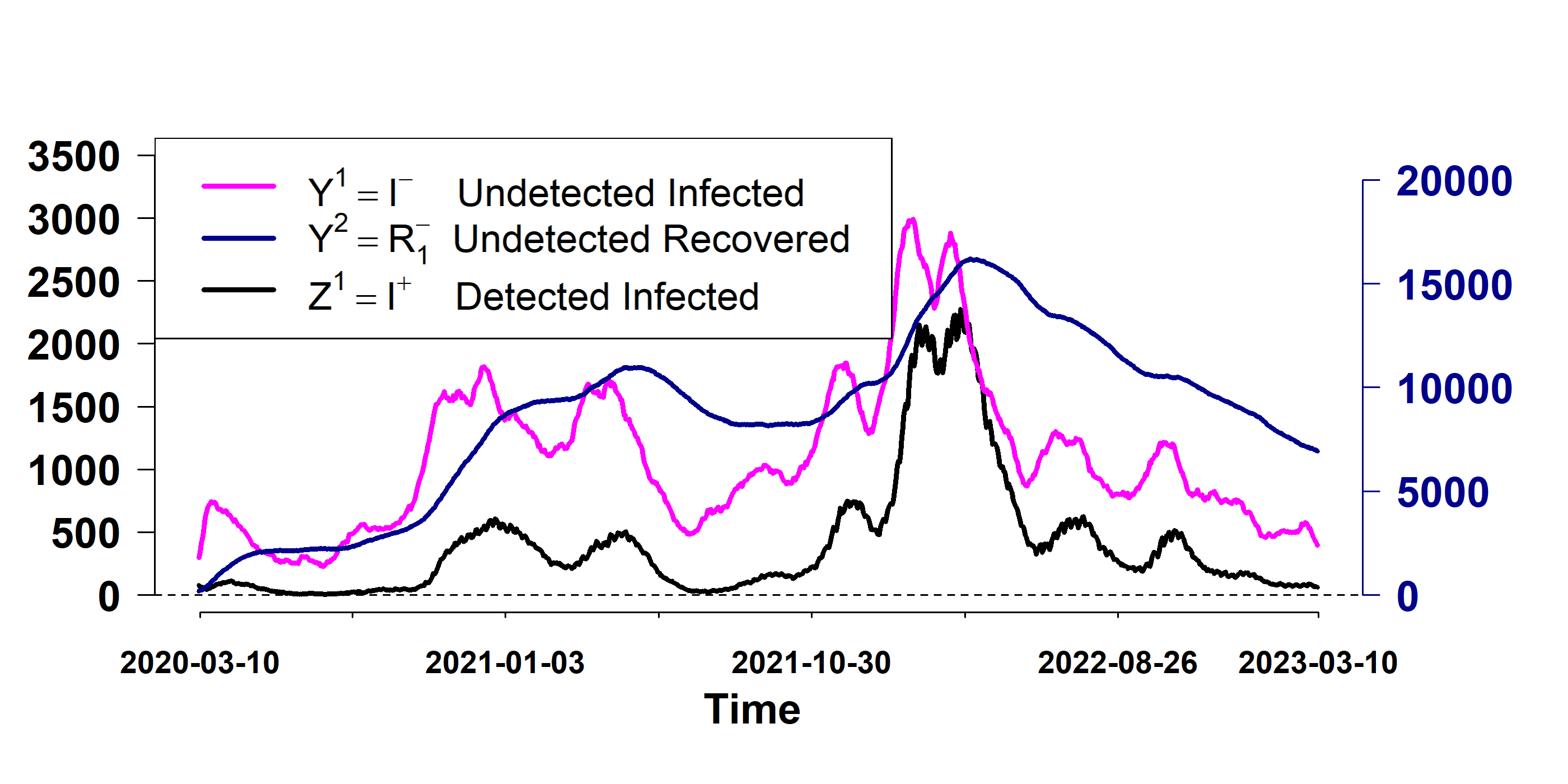}\label{fig:f171}}
	\hfill
	\subfloat[Hospitalized individuals, patients in ICU (left axis), and patients who died in ICU (right axis). ]{\includegraphics[width=0.48\textwidth,height=0.26\textwidth]{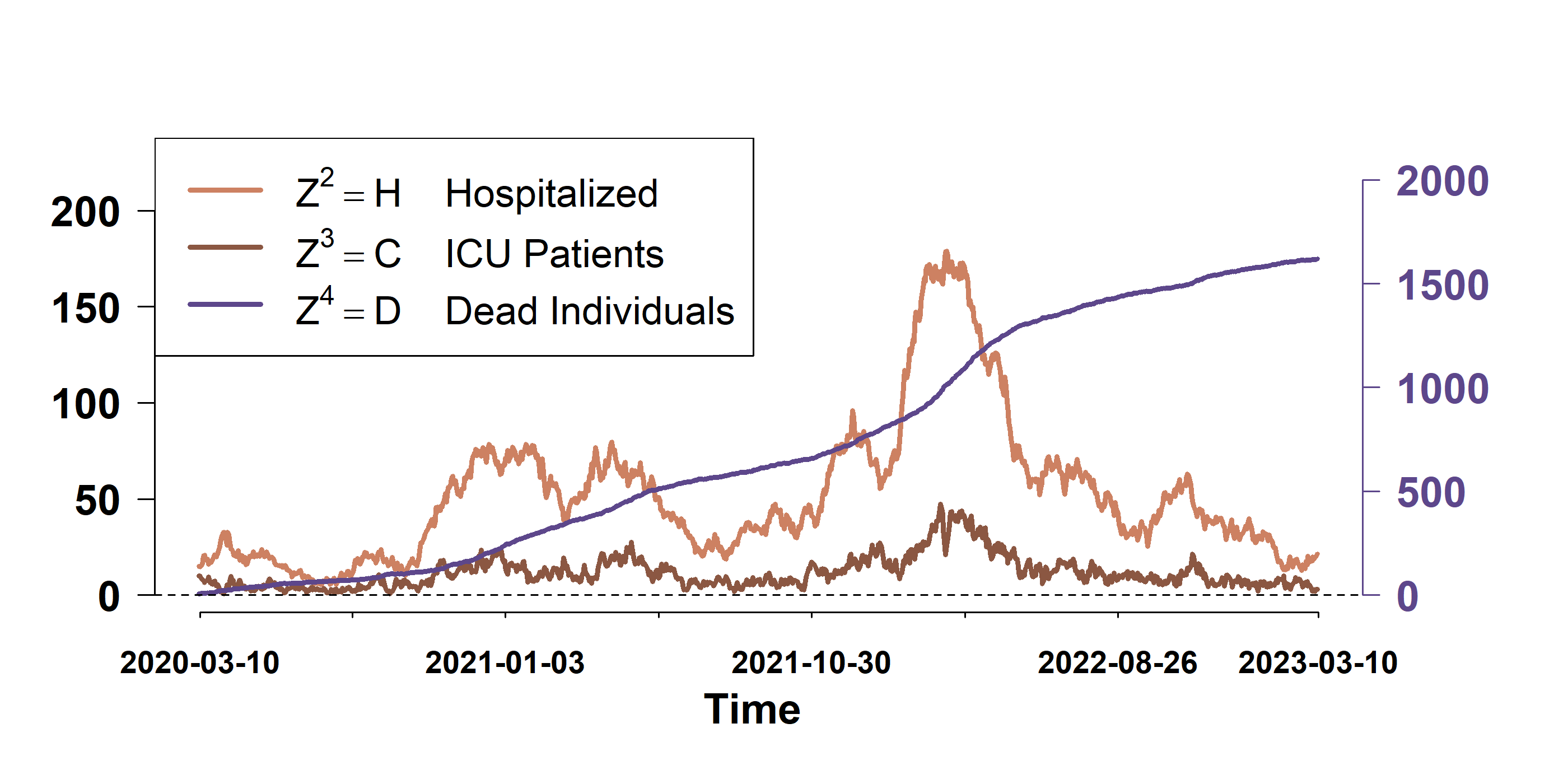}\label{fig:f172}}\\
	
	\subfloat[Observed recovered   $R^+_1,R^+_2,R^+_3$ (left axis).\\  Recovered with fading immunity (hidden) $R^-_2$  (right axis). ]{\includegraphics[width=0.48\textwidth,height=0.26\textwidth]{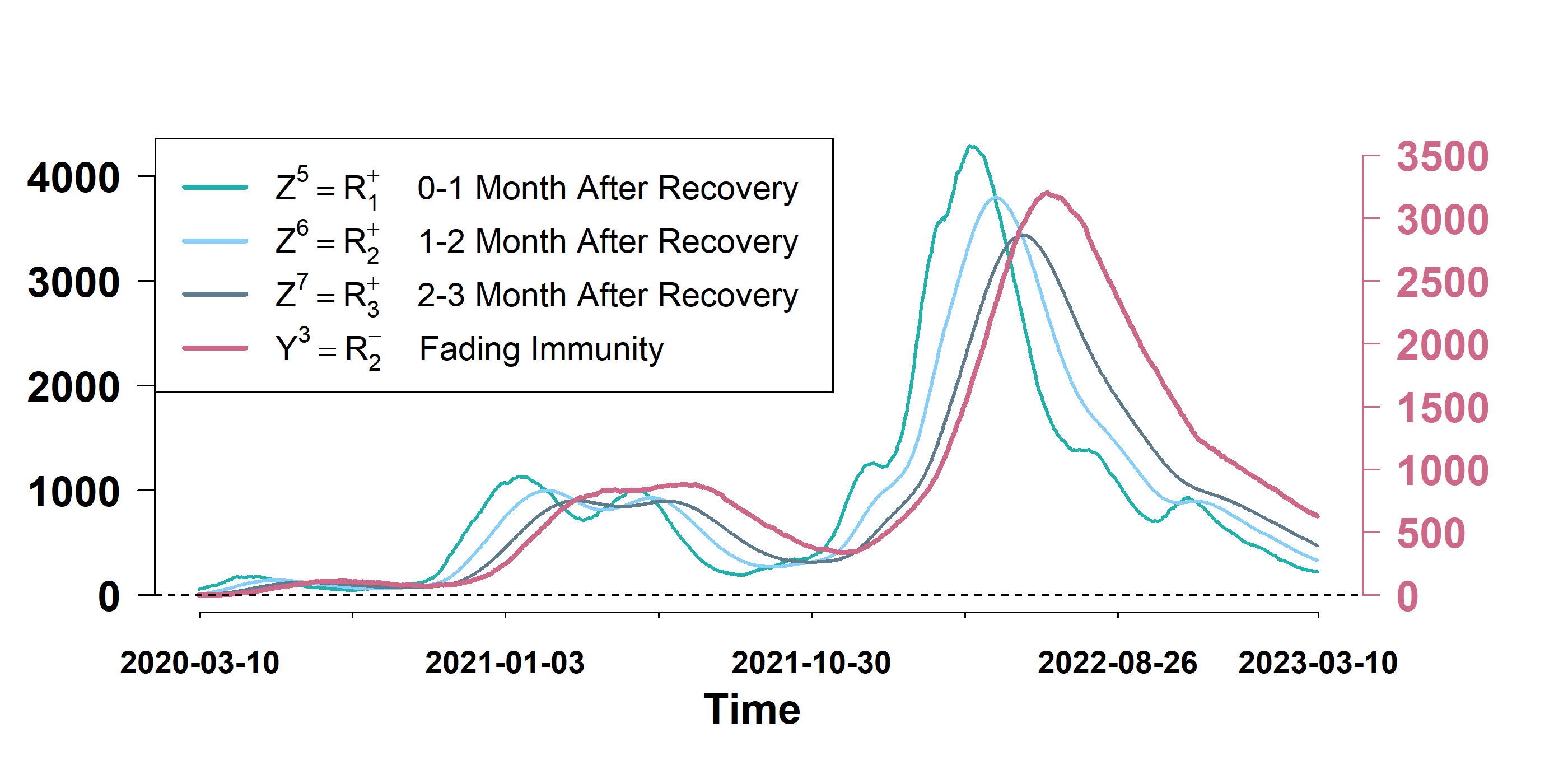}\label{fig:f173}}
	\hfill
	\subfloat[Observed vaccinated $V_1,V_2,V_3$ (left axis).\\ Vaccinated with fading immunity (hidden) $V^-$ (right axis). ]{\includegraphics[width=0.48\textwidth,height=0.26\textwidth]{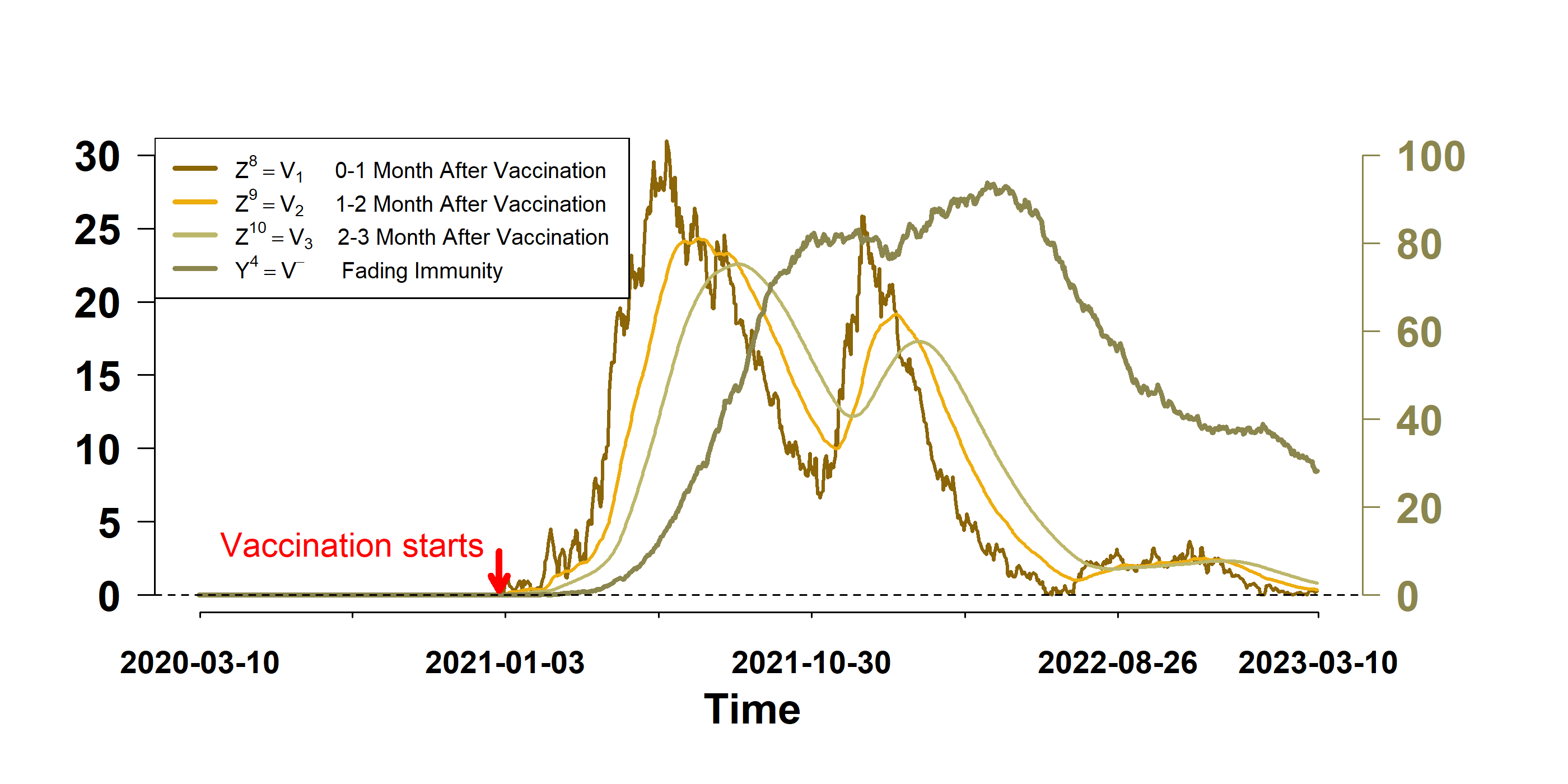}\label{fig:f174}}\\
	
	\caption{Simulated paths of hidden and observable states of the extended \covid model.}
	\label{fig:f170}
\end{figure}

We now apply the EKF Algorithm~\ref{Algo_EKF} to estimate the unobservable  states and focus on the performance evaluation of the  compartment $I^-$ of undetected infected which is the crucial state for monitoring the epidemic and  preventive action planning and resource allocation. 
Figure \ref{fig:f175}  displays  the true hidden state $I^-$, the filter estimate $\condmeanEKF_1=\condmeanEKF^{I^-}$, and the associated \(95\%\) confidence band, computed as \(\condmeanEKF^{1}_{n} \pm z_{0.975} (\condvarEKF^{11}_{n})^{1/2}\), where \(z_{0.975}\) denotes the quantile of order \(0.975\) of the standard normal distribution. 
\begin{figure}[h]
	\centering
	\includegraphics[width=1\textwidth]{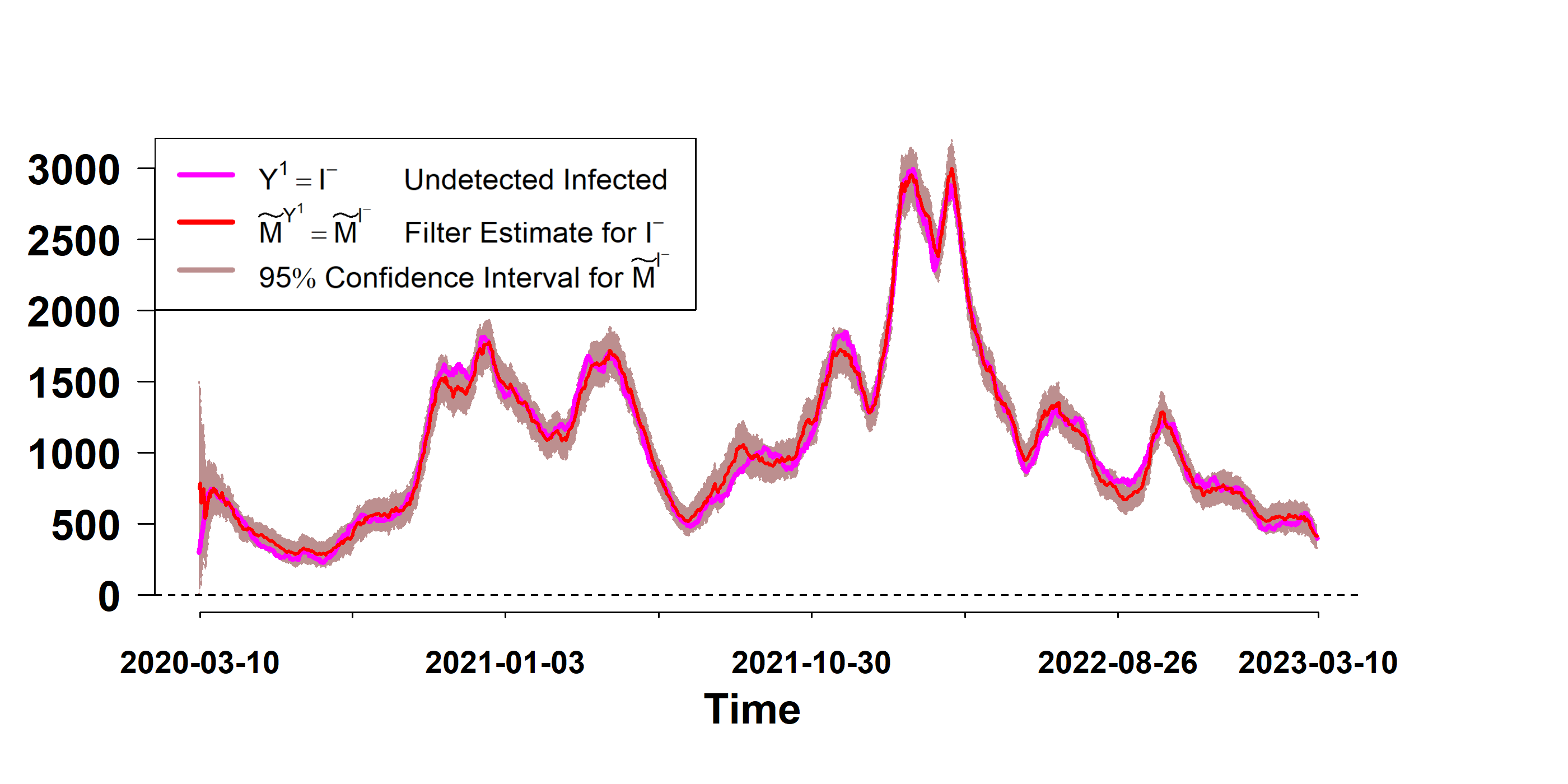}
	\caption{  True hidden state $I^-$, filter estimate $\condmeanEKF_1=\condmeanEKF^{I^-}$, and associated \(95\%\) confidence band showing that large initial uncertainty is reduced by learning from observations.
	}
	\label{fig:f175}
\end{figure}
It can be seen  that the filter estimate closely tracks the true signal, highlighting the accuracy and reliability of the EKF in capturing the underlying epidemic dynamics. This figure reveals that the initially rather large  estimation error diminishes rapidly as the filtering process begins. 
In addition, the confidence band provides insight into the accuracy of the filter estimates. At the beginning, there is a high degree of uncertainty, which is reflected in the considerable width of the confidence band. However, this confidence band narrows considerably within a few days, demonstrating the increasing accuracy of the filter as more observations are processed. After this initial “learning and warm-up phase,” the filter accuracy stabilizes, and the confidence interval reaches a relatively constant width. The above mentioned effects can also be seen from conditional standard deviations shown in Figure \ref{fig:f117}.

\begin{figure}[h]
	\centering
	\includegraphics[width=0.9\textwidth]{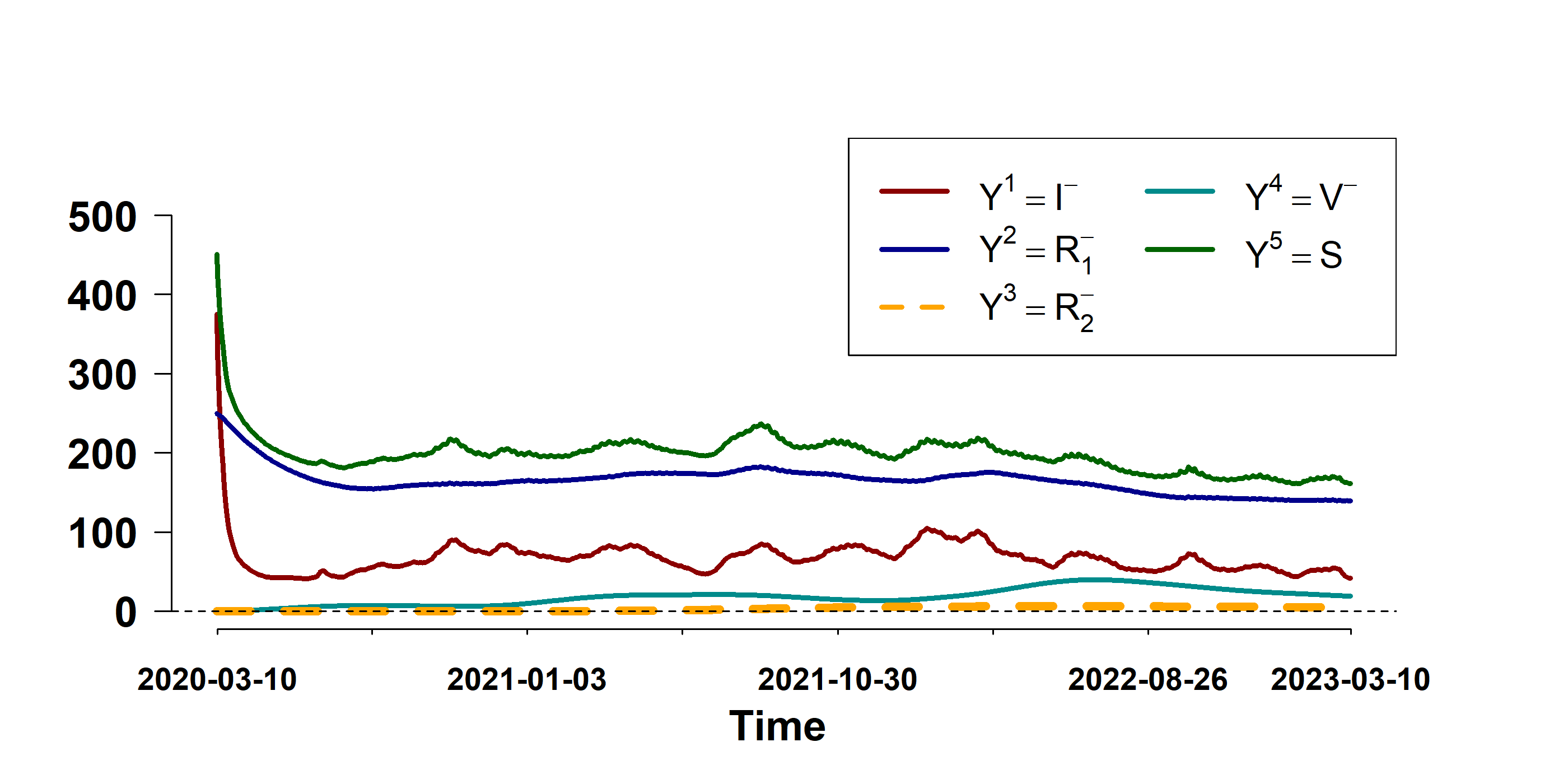}
\caption{  Conditional  standard deviations of the filter estimates of the five hidden states}	
	\label{fig:f117}
\end{figure}

\subsection{Impact of  Initial Estimates }\label{Impact_IEtim}
 We will now analyze the impact of initial estimates on the performance of the proposed filter method through a series of numerical experiments. To do this, we will vary the various initial estimates, namely the conditional mean $\condmeanEKF_0=\condmean_0$ and the conditional variance $ \condvarEKF_0=\condvar_0$, and observe the filter accuracy in the short and long term. Again, we focus on the results for the compartment $I^{-}$ of undetected infected  individuals. For the other hidden compartments, we observed similar results. 
	 
\begin{table}[h!]
	\centering
	\begin{tabular}{|c|c|c|c|}
		\hline
		&  & \multicolumn{2}{|c|}{Estimation} \\
		&True value & Conditional mean & Conditional variance\\
		& $I^{-}_{0}$ & $\condmean_{0}^{I^{-}}$ &   $\condvar_0^{I^-}$  \\ \hline
		Scenario 1            & $300$            & $750$  &$ 375^2$                    \\ \hline
		Scenario 2            & $300$            & $750$     & $0$                 \\ \hline
		Scenario 3            & $300$            & $300$    & $0$                 \\ \hline
	\end{tabular}
	\caption{Different scenarios  for the  initials estimates $\condmean_{0}^{I^{-}}$ and    $\condvar_0^{I^-}$}
	\label{tab:Val_Impact}
\end{table}

We assume that the prior information encoded in $\Fprior$ used to determine $\condmean_{0}^{I^-}$ and $\condvar_{0}^{I^-}$ stems from an expert or analyst who also knows the initial observation $Z_0$. The expert's views can be translated into  the parameters of the conditional distribution of $I^-_0 $ which assumed to be Gaussian.
Below, we distinguish between the following three scenarios in which we vary $\condmean_{0}^{I^-}$ and $\condvar_{0}^{I^-}$ as specified in Table \ref{tab:Val_Impact}, but keep the other initial estimates as shown above in Table \ref{tab:init_vals}. To visualize the results, we limit ourselves to the period of the first $75$ days after the start of the pandemic. We have observed that after this period, the influence of the initial estimates no longer plays a significant role.

\begin{figure}[h]
	\centering
	\subfloat[ Scenario 1 (red): Large initial uncertainty is reduced by learning from observations.
	Scenario 2 (blue):  Incorrectly specified initial estimate with perfect accuracy  needs long time to be corrected.
	]
	{\includegraphics[width=0.48\textwidth]{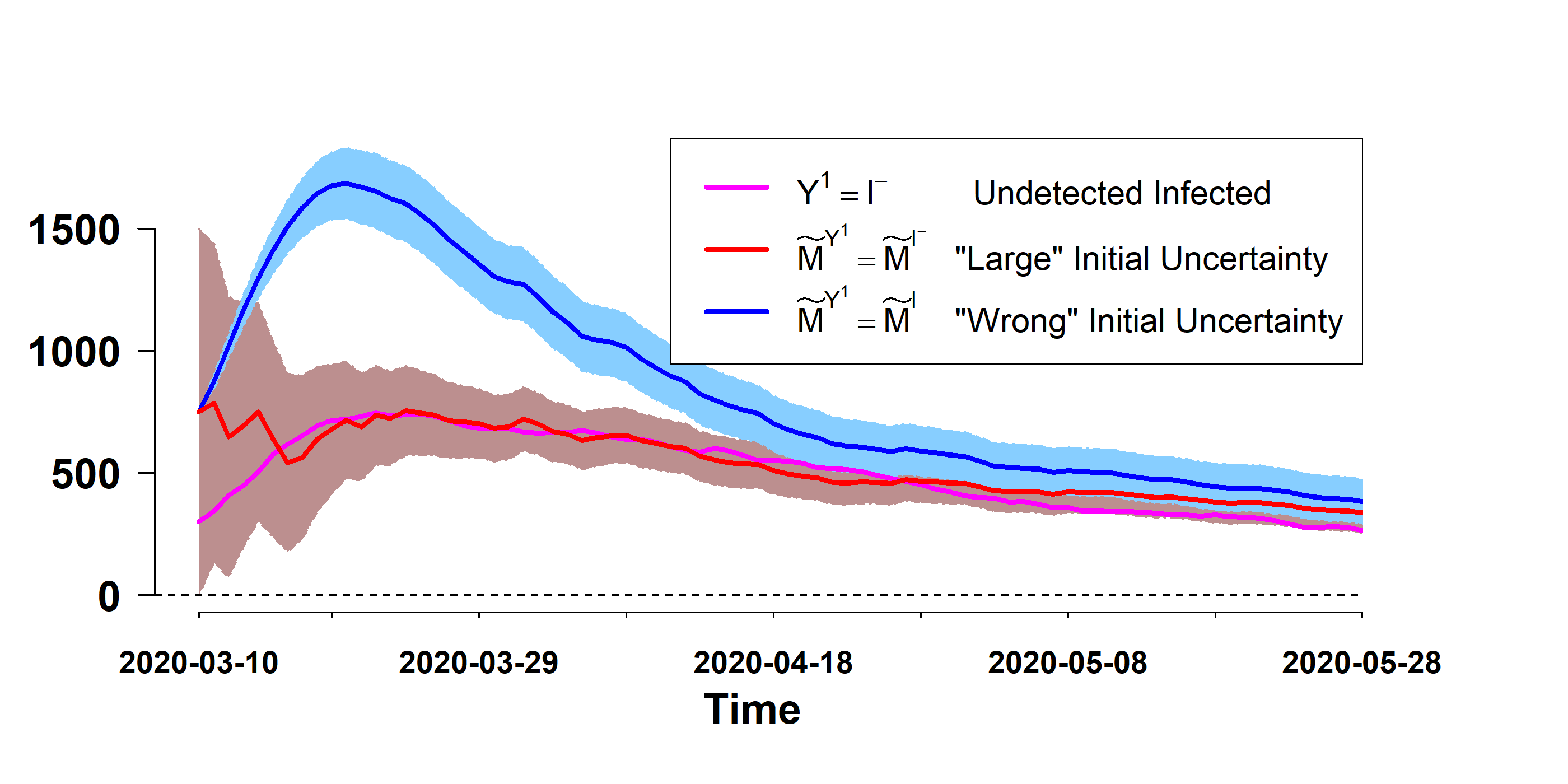}\label{fig:f179}}
	\hfill
	\subfloat[ Scenario 1 (red), scenario 3 (blue):  Zero initial uncertainty is fading out by observation noise. ]{\includegraphics[width=0.48\textwidth]{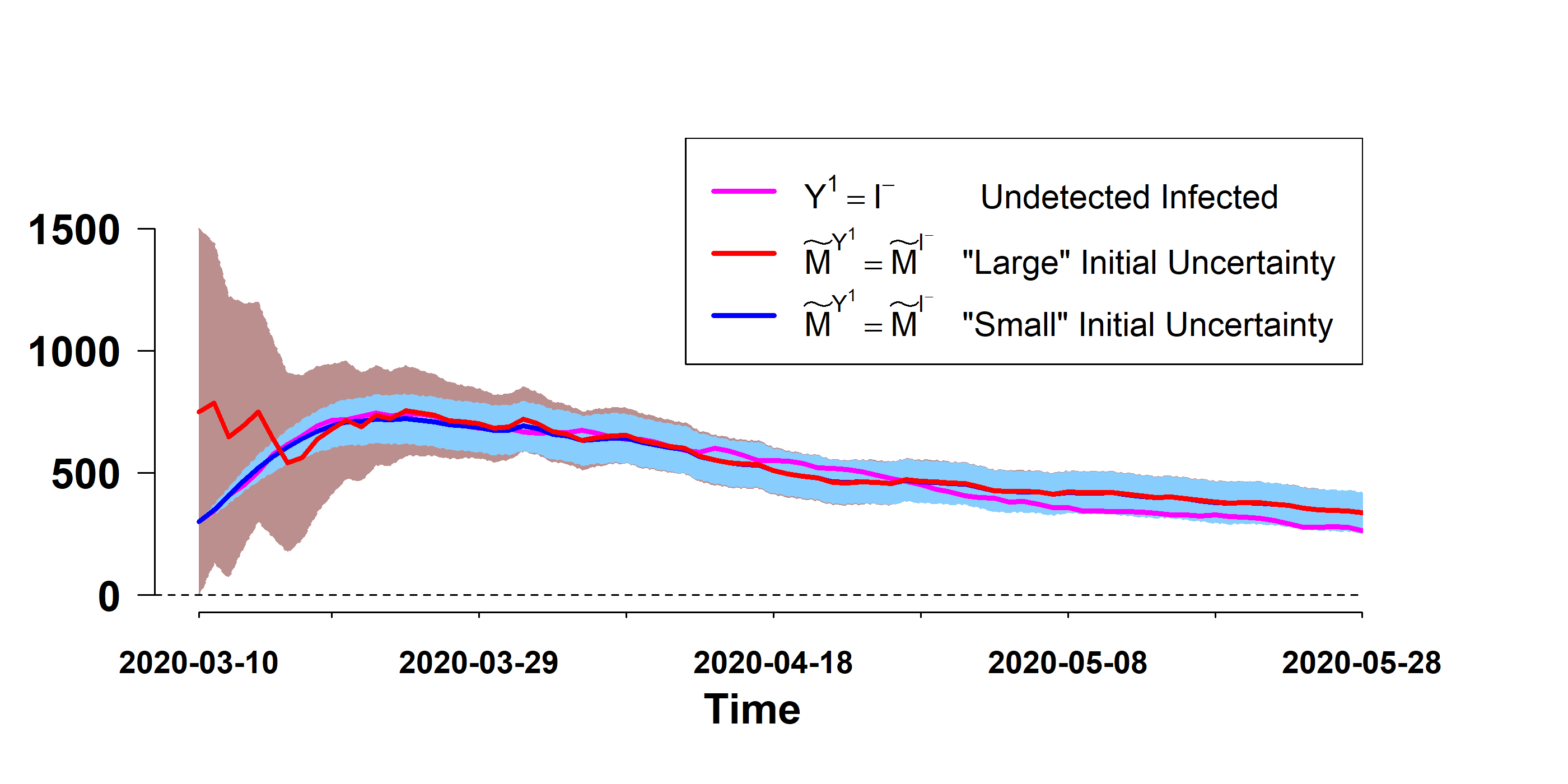}\label{fig:f180}}\\

	\caption{   
		Comparison of the effect of initial uncertainty on filtering performance. The figure shows for the three scenarios the true hidden state $I^-$, the filter estimate $\condmeanEKF_1=\condmeanEKF^{I^-}$, and the associated \(95\%\) confidence band.  }
	\label{fig:filter_initial}
\end{figure}

\paragraph{Scenario 1:  Poorly informed and uncertain expert}
This is the reference scenario and it coincides with the setting of the above subsection. 
The interpretation is that the expert gives a view in terms of the conditional distribution of the initial DFC $\dfc_{0}^{1}=I^-_0/I^+_0$ for which the mean $\dfcmean^1=10$ and variance $\dfcvar^1=25$ are specified. Since the expert is aware of the rather insufficient information required for an accurate  estimate,  a relatively large variance is specified. 

Figure \ref{fig:f179} shows what was already observed in Figure \ref{fig:f175} for the three-year period, namely that despite the relatively large initial estimation error, the filter takes about two weeks to correct itself, resulting in a fairly good accuracy compared to the actual signal. The confidence band (shown in light red), which was initially very wide (due to the large initial variance), becomes significantly narrower after a few days.

\paragraph{Scenario 2: Poorly informed but overconfident expert}
This scenario is formally  obtained from the first scenario by replacing the large DFC variance $\dfcvar^1$  with  zero. This corresponds to an overconfident expert which specifies the initial DFC distribution by an inaccurate mean $\dfcmean^1=10$, but at the same time  specifies an unrealistically perfect accuracy for the estimate.
The results are shown together with those for scenario 1 in Figure \ref{fig:f179}. In both scenarios, the filter process $\condmeanEKF^{I^-}$ starts with the same value. However, due to the incorrectly specified initial conditional variance, the filter process $\condmeanEKF^{I^-}$ now takes much longer to correct itself than in scenario 1. The confidence band (shown in light blue), which was very small (because of the almost vanishing initial variance) at first, becomes much more wide after a few days. This could be explained by the fact that, given the low variance values, the filter trusts these values and does not immediately make the correction, resulting in a longer adjustment. However, the  filter eventually adjusts and maintain fairly good accuracy over the long term.

\paragraph{Scenario 3:  Fully informed expert}
In this scenario the expert enjoys full information about the initial value $I^-_0$ at time $n=0$. That is, the view provides the accurate estimate  $\condmeanEKF^{I^-_0}=I^-_0$ together with a vanishing conditional variance $\condvarEKF^{I^-}_0$ indicating the perfect accuracy.
The results are shown in Figure \ref{fig:f180} together with those for scenario 1. We note that in this case, the signal $I^-$ must again be estimated solely from noisy observations after it was specified precisely at the initial point in time. This explains why the initially perfect accuracy deteriorates over time, leading to a progressive increase in the confidence interval (shown in light blue). Although the accuracy of the filter is no longer as high as at the beginning, it remains relatively accurate with respect to the actual signal.

\subsection{Impact of Cascade States}\label{Impact_CS}
\begin{figure}[h]
	\centering
	\subfloat[Standard deviation  \(\sqrt{\condvarEKF^{Y_{i}}}\) of hidden states over time of the estimation of $ I^- $ according to the different models. ]{\includegraphics[width=0.48\textwidth]{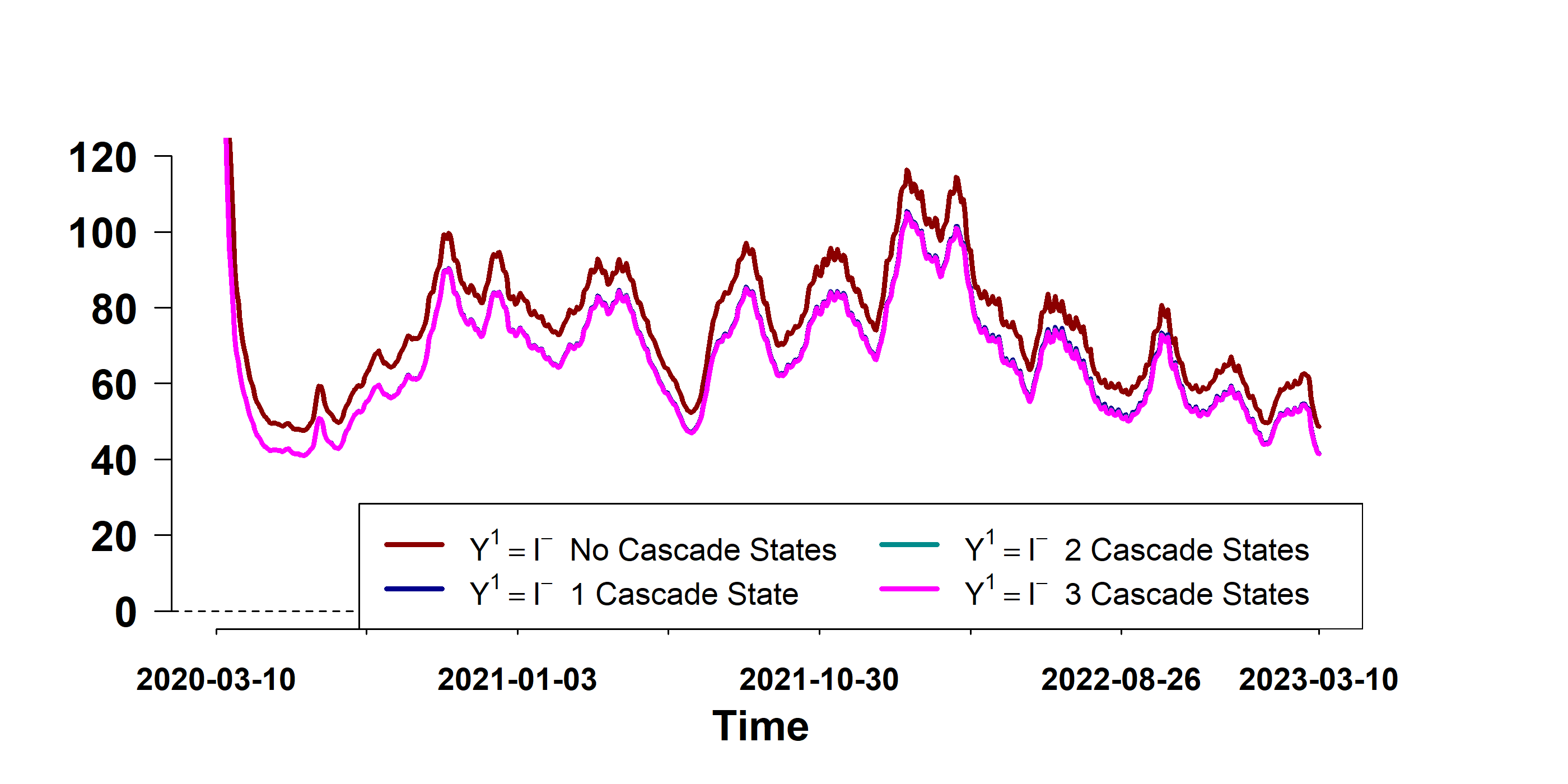}\label{chap33_2}}
	\hfill
	\subfloat[Zoom allowing the identification of all standard deviations \(\sqrt{\condvarEKF^{Y_{i}}}\). ]{\includegraphics[width=0.48\textwidth]{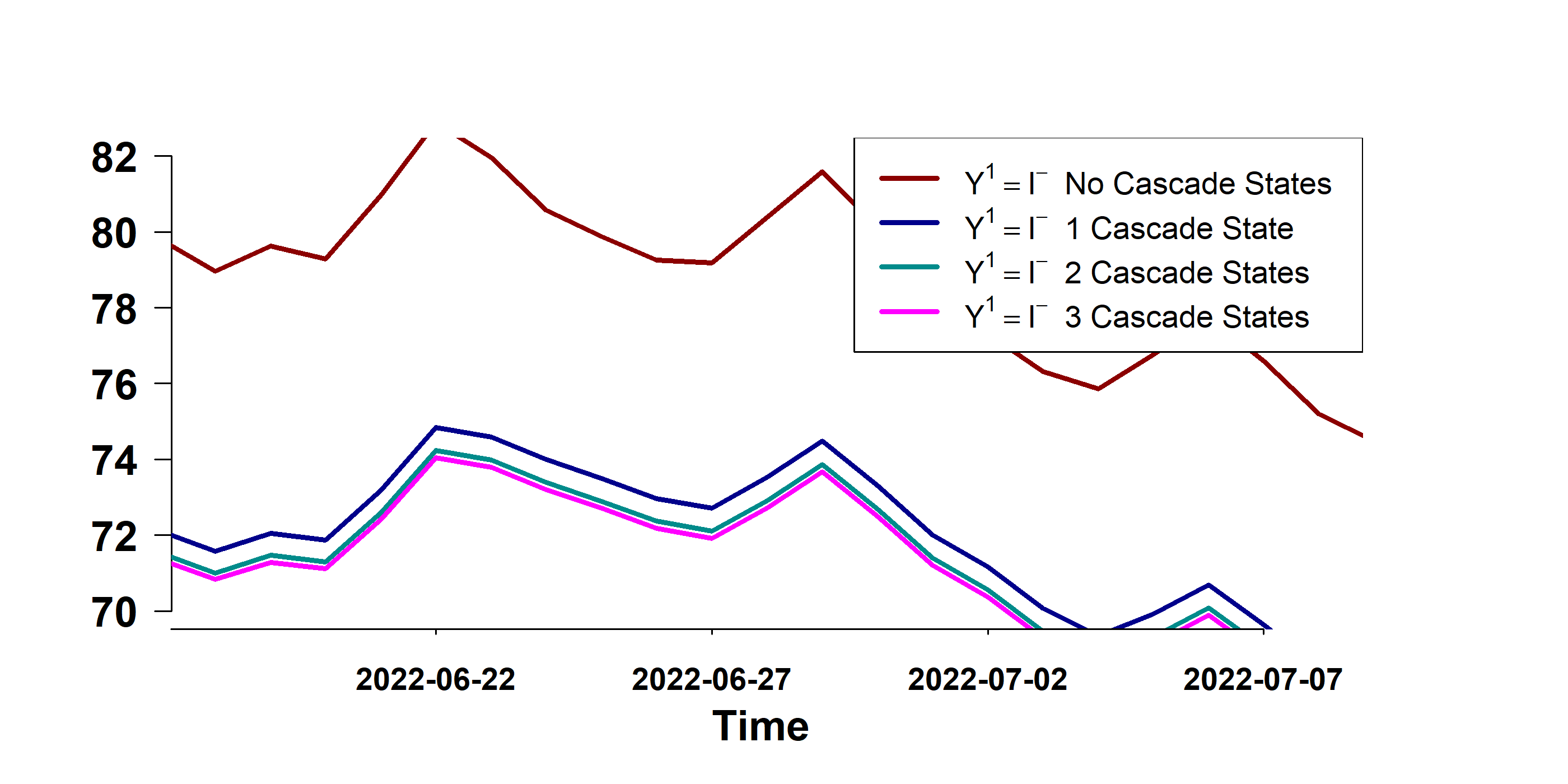}\label{chap33_3}}\\

	\caption{Standard deviation evolution over time }\label{chapp33}
\end{figure}

In this section, we assess how incorporating cascade compartments influences the performance of the filter, particularly with respect to estimation accuracy of the hidden states measured by the standard deviation. We compare the base model with the extended model equipped with one, two and three cascade compartments.
In both models, individuals who have recovered from a confirmed infection or have been vaccinated are classified as completely immune for $\LR=\LV=90$ days. In the base model they are  assigned to the compartments  $R^+$ and $V$, respectively. They are therefore not counted separately, but together with individuals for whom the time since recovery or vaccination exceeds $90$ days, and which can undergo an unobservable  transition  to the compartment  of susceptibles. Therefore, the compartments  $R^+$ and $V$ must be treated as hidden.
 In contrast, the extended \covid models includes  observable cascade compartments, allowing for a more detailed representation of the post-recovery and post-vaccination phases.

To ensure consistency, the same initial values are used across models for the hidden compartments $S,I^-$, and $R^{-}_{1}$ respectively $R^-$. For the observable recovered cascade compartments in the extended models, the initial value of  the first compartment $R^{+}_1$ is set to $ 50 $, while the others are initialized to zero. In the base model, the (partially) hidden compartment $R^+$ is also initialized to $ 50 $.
When cascade compartments are introduced, this the period of  $\LR=\LV=90$ days of complete immunity  is divided into $\KR = \KV = 1, 2,$ or $3$ subperiods, corresponding respectively to one compartment of 90 days, two compartments of 45 days each, or three compartments of 30 days each. 

 The immunity loss rates are chosen so that the average time spent in the immune state until immunity is lost and the transition to $S$ is the same for the basic and extended models. For immunity loss after recovery, we set $\rho^V=1/290$ in the basic model and, as above, $\rho^V=1/200$ in the extended model. Then the average duration of immunity after vaccination is $290=\LV+200$ in both cases.  For immunity loss after recovery, we set $\rho^{-}_{2} = 1/120$ in the base model and, as in the above experiments, $\rho^{-}_{2} = 1/30$ in the extended model. Here, too, the expected duration is the same and amounts to $120=\LR+30$.

Figure \ref{chap33_2} illustrates the evolution of the standard deviation over time for each model variant. As expected, the base model lacking additional observable compartments exhibits the highest variance. This highlights the critical role of incorporating observable states in enhancing estimation accuracy.
Interestingly, we observe that the majority of variance reduction is achieved by introducing the first observable cascade state. Adding a second or third cascade compartment yields diminishing returns. This likely stems from the fact that transitions between successive cascade states are deterministic; they do not introduce new randomness or observational information that can be exploited by the filter. To provide a more refined comparison, Figure \ref{chap33_3} offers a zoomed-in view of the standard deviations across models. The results are consistent with those obtained in earlier experiments.

\begin{appendix}	
\section*{Appendix}
\addcontentsline{toc}{section}{Appendix}
{\footnotesize

\section{Notation}
	\begin{longtable}{ll}
		\covid &  Coronavirus disease 2019 \\
		SDE &  Stochastic differential equation \\
		CTMC &    Continuous-time Markov chain \\
		ICU & Intensive care unit\\
		EKF &  Extended Kalmnan filter \\[0.5ex]
		$\alpha$ & Test rate   \\
		$\beta$ & Transmission rate \\
		$\gamma^\dagger$ &  Recovery rate, $\dagger=-,+,H,C$ \\
		$\delta$ &  Hospitalization rate in the ICU \\
		$\Delta t$ &  time step\\
		$\Count_{k}$ &    Counting process (for the number of transition $k$ in $[0,t]$)\\
		$\eta^{-},\eta^{+}$ &  Hospitalization rate for undetected / detected infected \\
		$\kappa$ & Death rate from ICU  \\
		$\lambda$ &  Intensity rate w.r.t absolute population sizes  $X$\\
		$\nu$ &   Intensity rate w.r.t. relative population sizes $\overline{X}$\\
		$\Pi$ &   Standard Poisson process with unit intensity\\		
		$\psiR,\psiV$ &    Weight factors in recursions of cascade state dynamics \\
		$\rho_^^ -,\rho_2^-,\rho^V$ &   Losing immunity rate\\		
		$\sigma$ &    Second diffusion coefficient, hidden state\\
		$\diffX$ & Diffusion coefficient of diffusion approximation $X^D$\\
		$\Noise_{n} ^1,\Noise_{n}^2$ &    Independent   $\mathcal{N}(0,\one)$ random vectors\\
		$C$ &  Individuals in the ICU \\
		$D$ &  Individuals  who died from the disaese in ICU \\
		$d,d_1,d_2$ &  Number of all, hidden and observable states  \\		   
		\ $\KR,\KV$ &  Number of cascade compartments \\
		$f, f_0,f_1$ &    Drift coefficient, hidden state\\
		  $\driftX$ &      Drift  coefficient of diffusion approximation $ X^D$\\			
		$g$ &    First diffusion coefficient, hidden state\\
		$H$ &   Hospitalized individuals \\
		$h, h_0,h_1$ &    Drift coefficient, observable state \\				
		$I^{-},I^{+}$ & undetected  / detected  infected \\
		$\one_{n}$ & $n\times n$ Identity matrix of order $n$ \\
		$K$ &  Total number of different transitions  \\
		$ \LR,\LV$ &  Number of time steps with complete immunity\\
		$\ell$ &    Diffusion coefficient, observable  state\\
		$\condmean,\condmeanEKF$ & Conditional mean/ EKF approximation\\
		$N$ &  Total population size \\
		$N_t$ &   Total number of time steps\\
		$\condvar,\condvarEKF$ & Conditional variance/ EKF approximation\\
		$\PR_j,\PV_j$ &   Number of  of original cascade compartments grouped to one compartment \\
		$R^{-},R^{+}$ & Undetected/detected recovered   in base model   \\
		$R^{-}_1,R^{-}_2$ & Undetected recovered/  detected recovered  with fading immunity  in extended model   \\
		$R_{j}^{+}$ & Cascade compartment with respect to observable recovered \\		
		$S$ &   Susceptible  \\		
		$T, t, t_0,\ldots,t_{N_t}$ &   Time horizon/ time/   discrete time points \\
		$\dfc$ &  Dark figure coefficient\\ 
		$V,V^{-}$ & Vaccinated/ vaccinated with fading immunity    \\
		$V_{j}^{}$ & Cascade compartment with respect to vaccinated individuals \\		
		$W, W^{1},W^{2}$ &   multi-dimensional standard Brownian motions\\
		$X,\overline{X}$ &  State vector for absolute and relative subpopulation size \\
		$Y,\widetilde{Y}$ & Hidden state in original/linearized system \\
		$\overline{Y}$ & Reference point for Taylor expansion \\
		$Z,\widetilde{Z}$ &  Observable state in original/linearized system \\		
	\end{longtable}

\section{Coefficients of the Recursion of the State Process}
\label{app:coeff}
\subsection{Base \covid  Model}
\label{app:base_coeff}
 Below we give the coefficients $f,h,\sigma,g,\ell$  appearing in the recursions \eqref{state_YZ} for the base model introduced in Subsection \ref{sec:BaseModel} with $d=9$ states, $K=16$ transitions,  the hidden state $Y=(I^-,R^-,R^+,V,S)^{\top}$ of dimension $d_1=5$, and the observable state $Z=(I^+,H,C,D)^{\top}$ of dimension $d_2=4$. The dimension of the Gaussian random vectors $\Noise^1,\Noise^2$  are $k_1=k_2=8$.

{\footnotesize
	\begin{align}
		\label{drift_signal_base}
		f(n,y,z)& = y+\left(\begin{array}{c}		
			\beta\frac{y^1y^5}{N} -(\alpha + \gamma^{+} + \eta^{-} - \mu )y^1\\
			\gamma^{-}y^1 -\mu y^2 -\rho^{-}_{1}y^2\\
			\gamma^{+}z^1 + \gamma^{H}z^2 + \gamma^{C}z^3 - \rho^{-}_{2}y^3\\
			(y^1 + y^2 + y^5)\mu - \rho^{V}y^4\\
			-\beta\frac{y^1y^5}{N} -\mu y^5 +\rho^{-}_{1}y^2 +\rho^{-}_{2}y^3 +\rho^{V}y^4
		\end{array} \right)\Delta t, ~~\\[1em]
		  h(n,y,z)& =  h_{0}(n,z) + h_{1}(n,z)y \quad \text{with}\\[1em]
		h_{0}(n,z)& =z+
		\left(\begin{array}{c}
			-\gamma^{+}z^1 - \eta^{+}z^1\\[0.7ex]
			\eta^{+}z^1 - \delta z^2 - \gamma^{H}z^2\\
			\delta z^2 - \gamma^{C}z^3 -\kappa z^3\\
			\kappa z^3
		\end{array} \right)\Delta t,		
		\quad
		h_{1}(n,z)
		= \left(\begin{array}{ccccc}
			\alpha & 0 & 0&0 &0\\
			\eta^{-} & 0 & 0 & 0 & 0\\
			0 & 0 & 0 & 0 & 0\\
			0 & 0 & 0 & 0 & 0
		\end{array} \right)	\Delta t	
		\\[1em] 
		\sigma(n,y,z)&=\left(\begin{array}{c@{\hspace*{-0.001cm}}c@{\hspace*{-0.001cm}}c@{\hspace*{-0.001cm}}c@{\hspace*{-0.001cm}}c@{\hspace*{-0.001cm}}c@{\hspace*{-0.001cm}}c@{\hspace*{-0.001cm}}c}			
			\sqrt{\beta\frac{y^1y^5}{N}} & -\sqrt{\gamma^{-}y^1} &  \sqrt{\mu y^1} & 0 & 0 & 0 & 0 & 0\\
			0 & \sqrt{\gamma^{-}y^1} & 0 & -\sqrt{\mu y^2} & -\sqrt{\rho^{-}_{1}y^2} & 0 & 0 & 0 \\
			0 & 0 & 0 & 0 & 0 & -\sqrt{\rho^{-}_{2}y^3} & 0 & 0\\
			0 & 0 & \sqrt{\mu y^1} & \sqrt{\mu y^2} & 0 & 0 & \sqrt{\mu y^5} & -\sqrt{\rho^{V} y^4} \\
			-\sqrt{\beta\frac{y^1y^5}{N}} & 0 & 0 & 0 & \sqrt{\rho^{-}_{1}y^2} & \sqrt{\rho^{-}_{2}y^3} & -\sqrt{\mu y^5} & \sqrt{\rho^{V}y^4} 
		\end{array} \right)\sqrt{\Delta t}\\[1em]
		g(n,y,z)&=\left(\begin{array}{c@{\hspace*{-0.001cm}}c@{\hspace*{-0.001cm}}c@{\hspace*{-0.001cm}}c@{\hspace*{-0.001cm}}cccc}		
			-\sqrt{\alpha y^1} & -\sqrt{\eta^{-}y^1} & 0 & 0 & 0 & 0 & 0 & 0\\
			0 & 0 & 0 & 0 & 0 & 0 & 0 & 0\\
			0 & 0 & \sqrt{\gamma^{+}z^1} & \sqrt{\gamma^{H}z^2} & \sqrt{\gamma^{C}z^3} & 0 & 0 & 0\\
			0 & 0 & 0 & 0 & 0 & 0 & 0 & 0\\
			0 & 0 & 0 & 0 & 0 & 0 & 0 & 0
		\end{array} \right)\sqrt{\Delta t}\\
		\ell(n,y,z)&=\left(\begin{array}{c@{\hspace*{-0.001cm}}c@{\hspace*{-0.001cm}}c@{\hspace*{-0.001cm}}c@{\hspace*{-0.001cm}}c@{\hspace*{-0.001cm}}c@{\hspace*{-0.001cm}}c@{\hspace*{-0.001cm}}c}	
			\sqrt{\alpha y^1} & 0 & -\sqrt{\gamma^{+}z^1} & 0 & 0 & -\sqrt{\eta^{+}z^1} & 0 & 0 \\
			0 & \sqrt{\eta^{-}y^1} & 0 & -\sqrt{\gamma^{H}z^2} & 0 & 0 & 0 & -\sqrt{\delta z^2}\\
			0 & 0 & 0 & 0 & -\sqrt{\gamma^{C}z^3} & 0 & -\sqrt{\kappa z^3} & \sqrt{\delta z^2}\\
			0 & 0 & 0 & 0 & 0 & 0 & \sqrt{\kappa z^3} & 0
		\end{array} \right)\sqrt{\Delta t}
	\end{align}
The coefficients $f_0,f_1$ appearing in 	Lemma \ref{lem:linearized_system} and resulting from the linearization of the signal drift coefficient $f$ given above in \eqref{drift_signal_base} are given by  
\begin{align}
\label{drift_signal_linear_base}
\begin{split}	
	f_0(n,y,z)& =
	\left(\begin{array}{c}		
		-\beta\frac{y^{1}y^{5}}{N} \\
		0\\
		\gamma^{+}z^{1} + \gamma^{H}z^{2}+\gamma^{C}z^{3} \\
		0 \\
		\beta\frac{y^{1}y^{5}}{N} \\
	\end{array} \right),
	\\[1ex]
	f_1(n,y,z)& = \one_{5} +\left(\begin{array}{ccccc@{\hspace*{-0.0em}}}			
		\beta\frac{y^{5}}{N} -(\alpha+\gamma^{-}+\eta^{-}+\mu)& 0 &   0 & 0 &  \beta\frac{y^{1}}{N} \\
		\gamma^{-} & -(\mu + \rho_{1}^{-}) &  0 & 0 &   0  \\
		0 & 0 &   -\rho^{-}_{2} & 0 &   0  \\
		\mu & \mu &   0 & -\rho^{V}  & \mu  \\
		-\beta\frac{y^{5}}{N} & \rho_{1}^{-} &   \rho_{2}^{-} & \rho^{V} & -\beta\frac{y^{1}}{N} -\mu
	\end{array} \right) .
\end{split}			
\end{align}

}
	
\subsection{Extended \covid Model}
\label{app:ext_coeff}	
 Below we give the coefficients $f,h,\sigma,g,\ell$  appearing in the recursions \eqref{state_YZ} for the extended model for $\KR=\KV=3$ cascade states introduced in Subsection \ref{sec:ExtendedModel} with $d=15$ states, $K=21$ transitions, among them are $\KR+\KV=6$ deterministic transitions related to the cascade states.   The hidden state $Y=(I^-, R^{-}_{1}, R^{-}_{2}, V^{-}, S)^\top $ is of dimension $d_1=5$, and the observable state  $Z=(I^+,H,C,D,R^+_1,R^+_2,R^+_3,V_1,V_2,V_3)^{\top}$ of dimension $d_2=10$. The dimension of the Gaussian random vectors $\Noise^1,\Noise^2$  are $k_1=5$, and $k_2=11$, respectively. 

{\footnotesize
	\begin{align}
		\label{drift_signal_ext}
		f(n,y,z)& =y+ \left(\begin{array}{c}		
			\beta\frac{y^{1}y^{5}}{N} -(\alpha + \gamma^{-} + \eta^{-} - \mu )y^{1}\\
			\gamma^{-}y^{1} -\mu y^{2} -\rho^{-}_{1}y^{2}\\
			- \rho^{-}_{2}y^{3}   \\
			- \rho^{V}y^{4} \\
			-\beta\frac{y^{1}y^{5}}{N} -\mu y^{5} +\rho^{-}_{1}y^{2} +\rho^{-}_{2}y^{3} +\rho^{V}y^{4}\\
		\end{array} \right)\Delta t + \left(\begin{array}{c}		
			0\\
			0\\
			\psi_{3}z^{7}  \\
			\psi_{3}z^{10} \\
			0\\
		\end{array} \right)\\[0.5ex]
		  h(n,y,z)& =  h_{0}(n,z) + h_{1}(n,z)y \quad \text{with}\\[0.5ex]	
		h_{0}(n,z)& = z+
		\left(\begin{array}{c}		
			-\eta^{+}z^{1} -\gamma^{+}z^{1}\\
			\eta^{+}z^{1}  -\delta z^{2} -\gamma^{H}z^{2}\\
			\delta z^{2} - \gamma^{C}z^{3} -\kappa z^{3}\\
			\kappa z^{3}\\
			\gamma^{+}z^{1} + \gamma^{H}z^{2} + \gamma^{C}z^{3}  \\
			0   \\
			0 \\
			0 \\
			0\\
			0 
		\end{array} \right)\Delta t + \left(\begin{array}{c}		
			0\\
			0\\
			0\\
			0\\
			-\psiR_{1}z^{5}\\
			\psiR_{1}z^{5}  -\psiR_{2}z^{6} \\
			\psiR_{2}z^{6}  -\psiR_{3}z^{7} \\
			-\psiV_{1}z^{8} \\
			\psiV_{1}z^{8}  -\psiV_{2}z^{9}\\
			\psiV_{2}z^{9}  -\psiV_{3}z^{10} 
		\end{array} \right), \quad    		
		h_{1}(n,z) 
		=\left(\begin{array}{ccccc@{\hspace*{-0.0em}}}			
			\alpha &  0 & 0 & 0 & 0  \\
			\eta^{-} & 0 & 0 & 0 & 0\\
			0 & 0 & 0 & 0 & 0 \\
			0 & 0 & 0 & 0 & 0\\
			0 & 0 & 0 &  0 & 0 \\
			0 & 0 & 0 &  0 & 0 \\
			0 & 0 & 0 & 0 & 0 \\
			\mu  & \mu  &0 &  0 & \mu\\
			0 & 0 & 0 & 0 & 0 \\
			0 & 0 & 0 &  0 & 0
		\end{array} \right)
		\Delta t \\[1ex]
		\sigma(n,y,z)&=\left(\begin{array}{ccccc}			
			\sqrt{\beta\frac{y^{1}y^{5}}{N}} & -\sqrt{\gamma^{-}y^{1}} &   0 & 0 &   0 \\
			0 & \sqrt{\gamma^{-}y^{1}} &   -\sqrt{\rho^{-}_{1}y^{2}} & 0 &   0  \\
			0 & 0 &   0 & -\sqrt{\rho^{-}_{2}y^{3}} &   0  \\
			0 & 0 &   0 & 0  & -\sqrt{\rho^{V} y^{4}}  \\
			-\sqrt{\beta\frac{y^{1}y^{5}}{N}} & 0 &   \sqrt{\rho^{-}_{1}y^{2}} & \sqrt{\rho^{-}_{2}y^{3}} & \sqrt{\rho^{V} y^{4}}  
		\end{array} \right)\sqrt{\Delta t} \\[1ex]
		g(n,y,z)& =\left(\begin{array}{ccccccccccc}			
			-\sqrt{\mu y^{1}}  &    0 & 0 & -\sqrt{\alpha y^{1}} & -\sqrt{\eta^{-}y^{1}} & 0 & 0 & 0 & 0 & 0 & 0\\
			0 & -\sqrt{\mu y^{2}} & 0  & 0 & 0 & 0 & 0 & 0 & 0 & 0 & 0\\
			0 & 0 &  0 & 0 & 0 & 0 & 0 & 0 & 0 & 0 & 0\\
			0 & 0 & 0 &  0 & 0 & 0 & 0 & 0 & 0 & 0 & 0\\
			0 & 0 & -\sqrt{\mu y^{5}} &  0 & 0 & 0 & 0 & 0 & 0 & 0 & 0 
		\end{array} \right)\sqrt{\Delta t}\\[1ex]
		\ell(n,y,z)&=\left(\begin{array}{c@{\hspace*{-0.01cm}}c@{\hspace*{-0.01cm}}c@{\hspace*{-0.09cm}}c@{\hspace*{-0.05cm}}c@{\hspace*{-0.09cm}}c@{\hspace*{-0.05cm}}c@{\hspace*{-0.09cm}}c@{\hspace*{-0.05cm}}c@{\hspace*{-0.09cm}}c@{\hspace*{-0.05cm}}c@{\hspace*{-0.09cm}}c@{\hspace*{-0.05cm}}c@{\hspace*{-0.09cm}}c@{\hspace*{-0.05cm}}c@{\hspace*{-0.09cm}}c@{\hspace*{-0.05cm}}c@{\hspace*{-0.09cm}}c@{\hspace*{-0.05cm}}c@{\hspace*{-0.09cm}}c}			
			0 &  0 & 0 & \sqrt{\alpha y^{1}} & 0 & -\sqrt{\gamma^{+}z^{1}} & 0 & 0 & -\sqrt{\eta^{+}z^{1}} & 0 & 0 \\
			0 & 0 & 0 & 0 & \sqrt{\eta^{-}y^{1}} & 0 & -\sqrt{\gamma^{H}z^{2}} & 0 & \sqrt{\gamma^{+}z^{1}} & 0 & -\sqrt{\delta z^{2}}\\
			0 & 0 & 0 & 0 & 0 & 0 & 0 & -\sqrt{\gamma^{C}z^{3}} & 0 & -\sqrt{\kappa z^{3}} & \sqrt{\delta z^{2}}\\
			0 & 0 & 0 & 0 & 0 & 0 & 0 & 0 & 0 & \sqrt{\kappa z^{3}} & 0\\
			0 & 0 & 0 &  0 & 0 & \sqrt{\gamma^{+}z^{1}} & \sqrt{\gamma^{H}z^{2}} & \sqrt{\gamma^{C}z^{3}} & 0 & 0 & 0\\
			0 & 0 & 0 &  0 & 0 & 0 & 0 & 0 & 0 & 0 & 0\\
			0 & 0 & 0 & 0 & 0 & 0 & 0 & 0 & 0 & 0 & 0\\
			\sqrt{\mu y^{1}} & \sqrt{\mu y^{2}} &\sqrt{\mu y^{5}} &  0 & 0 & 0 & 0 & 0 & 0 & 0 & 0\\
			0 & 0 & 0 & 0 & 0 & 0 & 0 & 0 & 0 & 0 & 0\\
			0 & 0 & 0 &  0 & 0 & 0 & 0 & 0 & 0 & 0 & 0
		\end{array} \right)\sqrt{\Delta t}
	\end{align}	
The coefficients $f_0,f_1$ appearing in 	Lemma \ref{lem:linearized_system} and resulting from the linearization of the signal drift coefficient $f$ given above in \eqref{drift_signal_ext} are given by 	
	\begin{align}
		\label{drift_signal_linear_ext}
		\begin{split}
			f_0(n,y,z)&=\left(\begin{array}{c}		
				-\beta\frac{y^{1}y^{5}}{N} \\
				0\\
				\psi_{3}z^{7}  \\
				\psi_{3}z^{10} \\
				\beta\frac{y^{1}y^{5}}{N} \\
			\end{array} \right),\\[1ex]
			f_1(n,y,z)&=  \one_{5} + \left(\begin{array}{ccccc@{\hspace*{-0.0em}}}			
				\beta\frac{y^{5}}{N} -(\alpha+\gamma^{-}+\eta^{-}-\mu)& 0 &   0 & 0 &  \beta\frac{y^{1}}{N} \\
				\gamma^{-} & -(\mu + \rho_{1}^{-}) &  0 & 0 &   0  \\
				0 & 0 &   -\rho^{-}_{2} & 0 &   0  \\
				0 & 0 &   0 & -\rho^{V}  & 0  \\
				-\beta\frac{y^{5}}{N} & \rho_{1}^{-} &   \rho_{2}^{-} & \rho^{V} & -\beta\frac{y^{1}}{N} -\mu
			\end{array} \right).
		\end{split}		
	\end{align}

}

\section{Proof  of Lemma \ref{lem:initial_estimate}}\label{Proof_IE}
\begin{proof}
 The assertion for the conditional means $\condmean^1_0,\condmean^2_0$ and the entries $\condvar_0^{ij} $, $i,j=1,2$, of the conditional covariance matrix    follow immediately  from Assumption \ref{ass:dfc} saying that the DFCs $\dfc^1$ and $\dfc^2$ are independent with distribution  $\mathcal{N}(\dfcmean^i,\dfcvar^i)$, $i=1,2$. Further, the assumption 	 $Y_{0}^{3} = R_{2,0}^{-}=0$ and $ Y_{0}^{4}= V^-_0 =0$ implies zero conditional means $\condmean^3_0,\condmean^4_0$ and zero entries in third and fourth rows and columns of $\condvar_0$.

Since  the total population size $N$ is assumed to be constant, the normalization property implies
\begin{align}\label{hY5}
Y_{0}^{5} = N - Y_{0}^{1} - Y_{0}^{2} - Y_{0}^{3} - Y_{0}^{4} - \sum_{i=1}^{10} Z_{0}^{i}.
\end{align}
Taking the conditional expectation given \( \mathcal{F}_{0}^{Z} \), we obtain
\[
\mathbb{E}[Y_{0}^{5} | \mathcal{F}_{0}^{Z}] = N - \mathbb{E}[Y_{0}^{1} | \mathcal{F}_{0}^{Z}] - \mathbb{E}[Y_{0}^{2} | \mathcal{F}_{0}^{Z}] - \sum_{i=1}^{10}Z_{0}^{i}.
\]
For the conditional covariance between \( Y_{0}^{1} \) and \( Y_{0}^{5} \) we use \eqref{hY5} and apply the bilinearity of the conditional covariance to obtain 
\begin{align}
\cov(Y_{0}^{1}, Y_{0}^{5} | \mathcal{F}_{0}^{Z}) &= \cov\Big(Y_{0}^{1}, N - Y_{0}^{1} - Y_{0}^{2} - \sum_{i=1}^{10} Z_{0}^{i} | \mathcal{F}_{0}^{Z}\Big)\\
&= -\cov(Y_{0}^{1}, Y_{0}^{1} | \mathcal{F}_{0}^{Z}) - \cov(Y_{0}^{1}, Y_{0}^{2} | \mathcal{F}_{0}^{Z}) - \sum_{i=1}^{10} \cov(Y_{0}^{1},Z_{0}^{i} | \mathcal{F}_{0}^{Z})\\
&=  -\var(Y_{0}^{1} | \mathcal{F}_{0}^{Z}) = -\dfcvar^1( Z_{0}^{1})^2,
\end{align}
where we have used the conditional independence of \( Y_{0}^{1} \) and \( Y_{0}^{2} \), and that $Z_{0}^{i}$ is  $\mathcal{F}_{0}^{Z}$-measurable, hence  $\cov(Y_{0}^{1},Z_{0}^{i} | \mathcal{F}_{0}^{Z})=0$.
The  expression for the conditional covariance between \( Y_{0}^{2} \) and \( Y_{0}^{5} \) follows analogously, that is 
$
\cov(Y_{0}^{2}, Y_{0}^{5} | \mathcal{F}_{0}^{Z}) = -\var(Y_{0}^{2} | \mathcal{F}_{0}^{Z}) = -\dfcvar^2( Z_{0}^{2})^2.
$

Finally, for the conditional variance of \( Y_{0}^{5} \) it follows from  \eqref{hY5} and the above results 
\[
\var(Y_{0}^{5} | \mathcal{F}_{0}^{Z}) = \var(Y_{0}^{1} | \mathcal{F}_{0}^{Z}) + \var(Y_{0}^{2} | \mathcal{F}_{0}^{Z}).
\]
Substituting $\var(Y_{0}^{i} | \mathcal{F}_{0}^{Z}) = \dfcvar^i( Z_{0}^{i})^2$, $i=1,2$,  gives 
$
\var(Y_{0}^{5} | \mathcal{F}_{0}^{Z}) = \dfcvar^1( Z_{0}^{1})^2 + \dfcvar^2( Z_{0}^{2})^2.
$   
\qed

\end{proof}
}

\end{appendix}	

\begin{acknowledgements}
	The authors thank      Olivier Menoukeu Pamen (University of Liverpool), and the collaborators within the DFG research project ``CESMO - Contain Epidemics with Stochastic Mixed-Integer Optimal Control'', in particular Gerd Wachsmuth,   Armin Fügenschuh, Markus Friedemann, Jesse Beisegel (BTU Cottbus--Senftenberg),	for insightful discussions and valuable suggestions that improved this paper.

	\smallskip\noindent
	\textbf{Funding~} 	
	The  authors gratefully acknowledge the  support by the Deutsche Forschungsgemeinschaft (DFG), award number 458468407,  and by the  German Academic Exchange Service (DAAD), award number 57417894.	
\end{acknowledgements}

\bibliographystyle{acm}


\end{document}